\pdfoutput=1
\documentclass[11pt,a4paper]{article}
\usepackage[tracking]{microtype}
\usepackage[ttscale=0.81]{libertine}
\usepackage[T1]{fontenc}
\usepackage[super]{natbib}
\setcitestyle{citesep={, }}

\usepackage{amsmath}
\usepackage{amsthm}
\usepackage{amssymb}
\usepackage{amsfonts}
\usepackage[top=1.0in, bottom=1.0in,left=1.0in, right=1.0in,marginparwidth=0.8in,marginparsep=3mm]{geometry}
\usepackage{fullpage}
\usepackage{graphicx,color}
\usepackage[dvipsnames]{xcolor}
\usepackage[utf8]{inputenc}
\usepackage{bm}
\renewcommand{\mathbf}[1]{\bm{#1}}

\usepackage{tikz}

\usepackage{enumerate}
\newcommand{\D}[2]{\ensuremath{\frac{\mathrm{d}#1}{\mathrm{d} #2}}}
\newcommand{\inbr}[1]{\ensuremath{\left\{#1\right\}}}
\newcommand{\inb}[1]{\ensuremath{\left[#1\right]}}
\newcommand{\inp}[1]{\ensuremath{\left(#1\right)}}

\usepackage{hyperref}
\hypersetup{
  colorlinks=true,
  linkcolor=green!30!black,
  linktoc=all,
  citecolor=blue,
  bookmarksnumbered=true
}

\newtheorem{theorem}{Theorem}[section]
\newtheorem{observation}[theorem]{Observation}
\newtheorem{corollary}[theorem]{Corollary}
\newtheorem{lemma}[theorem]{Lemma}
\newtheorem{proposition}[theorem]{Proposition}

\theoremstyle{definition}
\newtheorem{definition}[theorem]{Definition}

\newtheorem*{remarks}{Remarks}

\newcommand{\abs}[1]{\left\vert#1\right\vert}
\newcommand{\set}[1]{\left\{#1\right\}}
\newcommand{\eps}{\varepsilon}

\newcommand{\defeq}{:=}

\newcommand{\CommentS}[1]{}

\newcommand{\C}{\mathbb{C}}
\newcommand{\R}{\mathbb{R}}

\renewcommand{\S}{S_\eps}
\newcommand{\Sinv}{S_\eps^{\dagger}}

\newcommand{\fpk}{f_{\beta,k}^{\varphi}}
\newcommand{\fpsk}{f_{\beta,k,s}^{\varphi}}
\newcommand{\fpb}{f_{\beta',k}^{\varphi}}
\newcommand{\interior}[1]{\ensuremath{\mathrm{int}\inp{#1}}}

\renewcommand{\epsilon}{\varepsilon}

\newcommand{\st}[1][]{\ensuremath{\;{#1:}\;}}

\newcommand{\arxiv}[1]{\href{https://arxiv.org/abs/#1}{arXiv:#1}}
\renewcommand{\vec}[1]{\ensuremath{\mathbf{#1}}}

\newcommand{\tsaw}{\ensuremath{T_{\mathrm{SAW}}}}

\begin{document}
\title{Fisher zeros and correlation decay in the Ising model}
\author{Jingcheng Liu \and Alistair Sinclair \and Piyush Srivastava}

\let\bakthefootnote\thefootnote
\let\thefootnote\relax
\footnotetext{Jingcheng Liu, Computer Science Division, UC
 Berkeley. Email: \texttt{liuexp@berkeley.edu}. {Supported by US NSF grant CCF-1815328.}}

\footnotetext{Alistair Sinclair, Computer Science Division, UC
  Berkeley. Email: \texttt{sinclair@cs.berkeley.edu}. {Supported by US
    NSF grant CCF-1815328.}}

\footnotetext{Piyush Srivastava, Tata Institute of Fundamental
  Research, Mumbai. Email:
  \texttt{piyush.srivastava@tifr.res.in}. {Supported by a Ramanujan
    Fellowship of DST, India.}}

\let\thefootnote\bakthefootnote

\date{}
\maketitle
\begin{abstract}

  We study the complex zeros of the partition function
  of the Ising model, viewed as a polynomial in the ``interaction
  parameter''; these are known as \emph{Fisher zeros} in light of their
  introduction by Fisher in 1965.  While the zeros of the
  partition function as a polynomial in the ``field'' parameter have
  been extensively studied since the classical work of Lee and
  Yang, comparatively little is known about Fisher zeros for general graphs.  
  Our main result shows that the zero-field Ising model has
  no Fisher zeros in a complex neighborhood of the entire region of 
  parameters where the model exhibits correlation decay.  In addition
  to shedding light on Fisher zeros themselves, this result also establishes
  a formal connection between two distinct notions of phase transition
  for the Ising model: the absence of complex zeros (analyticity of the
  free energy, or the logarithm of the partition function) and decay of
  correlations with distance.
  We also discuss the consequences of our result for efficient deterministic
  approximation of the partition function.
  Our proof relies heavily on algorithmic techniques, notably Weitz's
  self-avoiding walk tree, and as such belongs to a growing
  body of work that uses algorithmic methods to resolve classical
  questions in statistical physics.

  \vfill
  
  \noindent \emph{\footnotesize An extended abstract of this paper
    comprising an announcement of the results has been accepted for
    presentation at the ``Innovations in Theoretical Computer Science
    (ITCS), 2019'' conference.}
\end{abstract}

\thispagestyle{empty}
\newpage
\setcounter{page}{1}
\section{Introduction}

The Ising model, which originated in the qualitative modeling of phase
transitions in magnets,\cite{isi25} was the first among the wide
class of \emph{spin systems} to be studied extensively in statistical
physics.  Given a graph $G = (V, E)$, the Ising model assigns to each
\emph{configuration} $\sigma: V \rightarrow \inb{+, -}$ of $\inb{+,-}$
\emph{spins} an energy
$H(\sigma) \defeq - \sum_{\inb{u, v} \in E} \sigma(u)\sigma(v)$.  The
weight of each configuration $\sigma$ is then $\exp(-JH(\sigma))$,
where $J$ denotes an \emph{inverse temperature} parameter.  The
\emph{partition function} of the model is given by
\begin{displaymath}
  Z_G(J) = \sum_{\sigma: V \rightarrow \inb{+,-}}\exp(-JH(\sigma)).
\end{displaymath}
The model naturally yields a probability distribution over the configurations, 
known as the \emph{Gibbs measure}, given by
$\mu_{G,J}(\sigma) = \frac{1}{Z_G(J)}\exp(-JH(\sigma))$. The setting
$J > 0$, in which neighboring vertices in the graph tend to have
similar spins, is called \emph{ferromagnetic}, while the setting
$J < 0$ is called \emph{anti-ferromagnetic}.

In this paper, we will find it more convenient to work with an
equivalent combinatorial view of the Ising model as a
probability distribution over the \emph{cuts} of the graph~$G$.  In this
view, one replaces the inverse temperature parameter $J$ with an
interaction parameter $\beta = \exp(-2J) > 0$, so that the weight $w(\sigma)$
assigned by the model to a configuration
$\sigma: V \rightarrow\inb{+, -}$ is given by
$$ w_{G,\beta}(\sigma) = \beta^{\abs{\inbr{e =(u, v) \in E \st
      \sigma(u) \neq \sigma(v)}}}.
$$  
Here, $\sigma$ is viewed as describing the cut in the graph between
spin-`$+$' and spin-`$-$' vertices.  As before, the associated
Gibbs measure assigns probability $\mu_{G, \beta}(\sigma) \defeq
\frac{1}{Z_G(\beta)}{w_{G,\beta}(\sigma)}$ to each
configuration~$\sigma$.  The normalizing factor here is the {\it
  partition function}, defined as
\begin{equation}
  Z_G(\beta) \defeq \sum_{\sigma : V \rightarrow \inbr{+, -}}
  w_{G,\beta}(\sigma) = \sum_{k=0}^{|E|} \gamma_k\beta^k,   \label{eq:1}
\end{equation}
where $\gamma_k$ is the number of $k$-edge cuts in~$G$.  Note that
$Z_G(\beta)$ is a polynomial in~$\beta$ with positive coefficients.
We also sometimes consider graphs in which certain vertices are
\emph{pinned} to `$+$' or `$-$' \emph{spins}.  For such a graph, we
restrict the sum in the definition of $Z_G$ to those configurations
$\sigma$ in which these vertices have the spin determined by their
pinning.

In physical terms, the
parameter $\beta$ above is a proxy for the 
``interaction strength'', while the graph is a proxy for the physical
structure of the magnet.  Further, in this parameterization,
$\beta > 1$ corresponds to so-called anti-ferromagnetic
interactions (where neighbors prefer to have different spins),
$\beta < 1$ to ferromagnetic interactions (where neighbors prefer to
have the same spins), and $\beta = 1$ to infinite temperature (where
the neighbors behave independently of each other).  We will restrict
our attention throughout to graphs of fixed maximum degree~$\Delta$
(i.e., a bounded number of neighbors per vertex).

Historically, there have been two distinct (though closely related)
mechanisms for defining and understanding phase transitions in
statisical physics.  The first is decay of long-range correlations in
the Gibbs measure.  The second, more classical mechanism is
analyticity of the ``free energy'' $\log Z$ (where $Z$ is the
partition function).  This second notion connects naturally to the
stability theory of polynomials, and in particular to the study of the
location of \emph{complex} roots of the partition function $Z$, even
when only real values of the parameters make physical sense in the
model.  The seminal work of Lee and Yang~\cite{lee_statistical_1952,
  leeyan52} was one of the first, and certainly the best known, to use
this notion.  It is interesting to note that the stability theory of
polynomials has seen a recent surge of interest following the central
role it has played in developments in a wide variety of areas ranging from
mathematical physics to combinatorics and theoretical computer
science:  examples include the resolution of the Kadison-Singer
conjecture \cite{marcus2015interlacing}, proofs of the existence of
Ramanujan graphs \cite{marcus_interlacing_2015}, and progress on
the traveling salesman problem and other algorithmic questions
(see, e.g.,
Refs.~\citenum{anari2014kadison,anari_generalization_2017,straszak_real_2016}).

\par\medskip\noindent
\paragraph{\textbf{Algorithms, phase transitions, and roots of
    polynomials.\ \ }} While the algorithmic consequences of phase
transitions defined in terms of decay of correlations have been well
studied, first in the context of Markov Chain Monte Carlo algorithms
(Glauber dynamics) and more recently in determinstic algorithms that
directly exploit correlation decay (see, e.g. Refs.~\citenum{Weitz,
  bandyopadhyay_counting_2008}), algorithmic use of the information on
complex roots of the partition function originated only recently in
the work of Barvinok (see Ref.~\citenum{barvinok2017combinatorics} for
a survey).  This has led to increased interest in understanding the
relationship between the above two notions of phase transitions.  Such
connections have been the focus of some recent work on the independent
set (or ``hard core lattice gas'') model; notably, connections similar
to the ones in this paper have been explored for that model by Peters
and Regts~\cite{peters17:_sokal}, while related ideas are harnessed in
early work of Shearer~\cite{Shearer}, as later elucidated by Scott and
Sokal~\cite{ScottSokal} and further elaborated by Harvey {\it et
  al.}~\cite{harvey2016computing}, to shed light on the Lov\'asz Local
Lemma.

The motivation for our work here is to take a step towards achieving a
fuller understanding of these connections.  Specifically, we study the
zeros of the Ising partition function (at zero field), viewed as
a polynomial in the interaction parameter.  While the study of
zeros in terms of the fugacity (or field) parameter was famously
pioneered by Lee and Yang~\cite{lee_statistical_1952}, and has given
rise to a well developed theory, very little is known about the zeros
in terms of the interaction parameter, which were first studied in the
classical 1965 paper of Fisher~\cite{fisher65} and are thus known
as ``Fisher zeros''.

Our main result is that the Ising model has no Fisher zeros in a region of the
complex plane that contains the entire interval~$B$ on the positive
real line where correlation decay holds.  Our analysis crucially
exploits the correlation decay property (see, in particular,
\MakeUppercase Proposition\nobreakspace \ref{lem:correlation-decay}
and\nobreakspace Lemma\nobreakspace \ref{lem:strict-contraction}) in
order to understand the Fisher zeros.  Thus, in the particular case of
the zero field Ising model, we are able to establish a tight
connection between correlation decay and the absence of zeros.
Another potentially interesting aspect of this result is the use of
algorithmic techniques associated with correlation decay (notably,
Weitz's algorithm~\cite{Weitz}) to understand a classical concept in
statistical physics.

We now proceed to formally describe our results.  First we identify
the range of the parameter~$\beta$ for which the Ising model is,
in a certain sense, well-behaved on graphs of bounded degree~$\Delta$.
\newcommand{\betainv}{\ensuremath{(\frac{\Delta - 2}{\Delta}, 1) \cup
    (1, \frac{\Delta }{\Delta - 2})}}
\newcommand{\betarange}{\ensuremath{(\frac{\Delta - 2}{\Delta},
    \frac{\Delta }{\Delta - 2})}}
\begin{definition}[\textbf{Correlation decay region}]
  Given $\Delta> 0$, the \emph{correlation decay region} $B  =
  B_\Delta$ for~$\beta$ is the interval $\betarange$.
\end{definition}
The correlation decay region is very well
studied in both physical and algorithmic contexts, and comes from a
consideration of the behavior of the Gibbs measure on trees.
In particular, it corresponds to those $\beta$
for which there is exponential decay of correlations in the Gibbs
measure on any finite subtree of the infinite $\Delta$-regular tree---a
fact which has been used to give a deterministic algorithm for
approximating the partition function of the Ising model for such
$\beta$.\cite{Weitz,zhaliabai09}  On the other hand, Sly and
Sun~\cite{sly12} have shown that for
$\beta > \frac{\Delta}{\Delta-2}$, this approximation problem is
NP-hard under randomized reductions.  In statistical physics, the
correlation decay region describes those $\beta$ for which the
definition of the Gibbs measure given by eq.~\eqref{eq:1} for finite
graphs can be extended in a unique way to a Gibbs measure on the
\emph{infinite} $\Delta$-regular
tree\cite{georgii88:_gibbs_measur_phase_trans}; for this reason, the
correlation decay region is also referred to as the \emph{uniqueness
  region}.

As advertised earlier, our goal is to prove the existence of a region of the
complex plane, containing~$B$, which contains no Fisher zeros.  We state this
now as our main theorem.
\begin{theorem}
  Fix any $\Delta > 0$. For any real
  $\beta \in B:=\bigl(\frac{\Delta - 2}{\Delta}, \frac{\Delta}{\Delta -
      2}\bigr)$, there exists a $\delta>0$ such that for all $\beta'\in \C$
  with $\abs{\beta' - \beta} < \delta$, the Ising partition
  function $Z_{G}(\beta') \neq 0$ for all graphs $G$ of maximum degree
  $\Delta$.  Moreover the same holds even if $G$ contains an arbitrary
  number of vertices pinned to $+$ or $-$ spins.
	\label{thm:main-fisher}
\end{theorem}

\begin{remarks}
  {\bf (1)}\ \ It is worth noting that the choice of $\delta$ does not
    depend on the size of the graph, only on $\Delta$ and~$\beta$.  
    In particular, given any $\delta_1 > 0$, one can choose
    $\delta > 0$ such that, for all $\beta'$ in a complex neighborhood
    of radius $\delta$ around the closed interval
    $[\frac{\Delta-2}{\Delta} + \delta_1, \frac{\Delta}{\Delta-2} -
    \delta_1]$, $Z_G(\beta')$ is non-zero for all graphs of degree at
    most $\Delta$. 
\par\smallskip\noindent
  {\bf (2)}\ \ For the case of the Ising model, the above theorem establishes
    a connection between the two notions of phase transition discussed above.
    Namely, for the zero-field Ising model, it shows that decay of correlations on the
    $\Delta$-regular tree also implies the
    absence of Fisher zeros for finite graphs of degree at most~$\Delta$,
    and hence the analyticity of the free energy for appropriate
    infinite graphs (i.e., those of maximum degree at most $\Delta$ and
    of subexponential growth, such as regular
    lattices).
\end{remarks}
\par\medskip\noindent
\paragraph{\textbf{Discussion.\ \ }}
While there are some results in the literature on Fisher zeros in
the case of specific regular lattices~(see, e.g.,
Refs.~\citenum{lu_density_2001} and \citenum{kim_partition_2008}), to
the best of our knowledge, the previous best general result on the
Fisher zeros of the Ising model appears in the work of Barvinok and
Sober\'{o}n~\cite{BarvinokSoberon16a}, who showed that $Z_G(\beta)$ is
non-zero if $\abs{\beta - 1}< c/\Delta$, where $\Delta$ is the maximum
degree of $G$, and $c$ can be chosen to be $0.34$ (and as large as
$0.45$ if $\Delta$ is large enough).  While this result provides a
disk around $1$ in which there are no Fisher zeros, it cannot
guarantee the absence of Fisher zeros in a neighborhood of the
correlation decay region~$B$ (which would require at least that
$c \geq 2 - o_\Delta(1)$).  Our Theorem~\ref{thm:main-fisher}
therefore strengthens this result to a neighborhood of the entire
correlation decay region~$B$.\footnote{Technically the results are
  incomparable in the sense that, while our results cover a much
  larger portion of the real line than that in
  Ref.~\citenum{BarvinokSoberon16a}, the diameter of the disk centered
  around $1$ in the region of Ref.~\citenum{BarvinokSoberon16a} may be
  larger than the radius guaranteed by our result.}

Our main theorem on Fisher zeros can also be combined with the
techniques of Barvinok~\cite{barvinok2017combinatorics} and Patel and
Regts~\cite{patel_deterministic_2016} to give a new deterministic
polynomial time approximation algorithm for the partition function of
the ferromagnetic Ising model with zero field on graphs of degree at
most $\Delta$ when $\beta \in \betarange$.  In particular, 
combining Theorem~\ref{thm:main-fisher} with Lemmas 2.2.1 and 2.2.3 of
Ref.~\citenum{barvinok2017combinatorics} (see also the discussion at the
bottom of page 27 therein) and the proof of Theorem 6.1 of
Ref.~\citenum{patel_deterministic_2016}, we obtain the following corollary:
\begin{corollary}
  Fix a positive integer $\Delta$ and $\delta > 0$.  There exist
  positive constants $\delta_1 > 0$ and $c$ such that for any complex
  $\beta$ with
  $\Re(\beta) \in \bigl[\frac{\Delta-2}{\Delta} + \delta,
    \frac{\Delta}{\Delta-2} - \delta\bigr]$ and
  $\abs{\Im(\beta)} \leq \delta_1$, the following is true.  There
  exists an algorithm which, on input a graph $G$ of degree at most
  $\Delta$ on $n$ vertices, and an accuracy parameter $\epsilon > 0$,
  runs in time $O(n/\epsilon)^c$ and outputs $\hat{Z}$ satisfying
  $\bigl|\hat{Z} - Z_G(\beta)\bigr| \leq \epsilon \abs{Z_G(\beta)}$.
\end{corollary}

For real $\beta$ in the same range, a deterministic algorithm with the
above properties, based on correlation decay,
was already analyzed in Ref.~\citenum{zhaliabai09}.
However, our extension to complex values of the parameter is of independent
algorithmic interest in light of the fact that algorithms for approximating the
Ising partition function at complex values of the parameters have
applications to the classical simulation of restricted models of quantum
computation.\cite{mann_approximation_2018}

Finally, we emphasize that in contrast to most other recent
applications of Barvinok's method (e.g.,
Refs.~\citenum{patel_deterministic_2016,BarvinokSoberon16b,BarvinokSoberon16a,barvinok_computing_2015,liu2017ising-full}),
where the required results on the location of the roots of the
associated partition function are derived without reference to
correlation decay, the algorithmic version of correlation decay is
crucial to our proof.  Indeed, implicit in our proof is an analysis of
Weitz's celebrated correlation decay algorithm~\cite{Weitz} (proposed
originally for the independent set, or ``hard core'', model, and
analyzed by Zhang, Liang and Bai~\cite{zhaliabai09} for the Ising
model in the case of real positive $\beta\in B$) for the Ising model
with \emph{complex}~$\beta'$ close to $\beta\in B$.  Thus, as
mentioned earlier, our work shows that Weitz's algorithm can be viewed
as a bridge between the ``decay of correlations'' and ``analyticity of
free energy'' views of phase transitions.  We note also that our work
is close in spirit to recent work of Peters and
Regts~\cite{peters17:_sokal} (see also Ref.~\citenum{bencs_note_2018}), who
employ correlation decay in the hard core model to prove stability
results for the hard core partition function.

\section{Outline of proof}
\label{sec:outline-proof}
We fix $\Delta$ to be the maximum degree throughout, and let $d = \Delta-1$.
Let $G$ be any graph of maximum degree~$\Delta$.  Our starting point is a
recursive criterion that guarantees that the partition function $Z_G(\beta)$ has no zeros.
For any non-isolated vertex~$v$ of~$G$, let $Z_{G, v}^+(\beta)$
(respectively, $Z_{G, v}^-(\beta)$)
be the contribution to $Z_{G}(\beta)$ from configurations with
$\sigma(v)=+$ (respectively, with $\sigma(v)=-$), so that
$Z_{G}(\beta) = Z_{G, v}^+(\beta) + Z_{G, v}^-(\beta)$.
Define also the ratio $R_{G,v}(\beta) := \frac{Z_{G,
    v}^+(\beta)}{Z_{G, v}^-(\beta)}$.
Now note that $Z_{G, v}^+(\beta)$ and $Z_{G, v}^-(\beta)$ can be seen as Ising partition
functions defined on the same graph $G$ with the vertex $v$
\emph{pinned} to the appropriate spin; i.e., they are partition functions 
defined on a graph with one less unpinned vertex.  Thus we may assume
recursively that neither $Z_{G,v}^+(\beta)$ nor $Z_{G, v}^-(\beta)$ vanishes.
Under this assumption, the condition $Z_{G}(\beta) \neq 0 $ is equivalent
to $R_{G, v}(\beta) \neq -1$.

Our next ingredient is a formal recurrence, due to Weitz\cite{Weitz},
for computing ratios such as $R_{G, v}(\beta)$ in two-state spin systems.
This recurrence is based on the so-called ``tree of self-avoiding walks'' (or 
``SAW tree'') in~$G$, rooted at~$v$, with appropriate boundary conditions
(i.e., initial inputs, or fixed values at the leaves of the tree).  Weitz's recurrence
has been used in the
development of several approximate counting algorithms based on decay
of correlations (see, e.g., Refs.~\citenum{Weitz, zhaliabai09,
  li_correlation_2011, sinclair_approximation_2012}).  We now state a 
  precise version of Weitz's result that is tailored to our application.
\begin{lemma}
  Let $G$ be a graph of maximum degree $\Delta = d + 1$, with some
  vertices possibly pinned to spins `$+$' or `$-$'.  Given
  $\beta \in \C$, define
  $h_\beta(x) \defeq \frac{\beta + x}{\beta x + 1}$.  For integers
  $k \geq 0$ and $s$, define the maps
  \begin{displaymath}
    F_{\beta, k, s}(\vec{x})  \defeq \beta^s \prod_{i=1}^k h_\beta(x_i).
  \end{displaymath}
  Then, the ratio $R_{G, v}(\beta)$ can be obtained by iteratively
  applying a sequence of multivariate maps of the form
  $F_{\beta, k, s}(\vec{x})$ such that, in all but the final application,
  one has $1 \leq k + \abs{s} \leq d$, while for the final application
  one has $1 \leq k + \abs{s} \leq \Delta$, and any
  initial input to these maps is $x_i=1$.
  \label{lem:weitz}
\end{lemma}
\noindent For completeness we sketch a proof of Lemma~\ref{lem:weitz} at the end of this section.

Returning now to the condition $R_{G, v}(\beta) \neq -1$ derived above,
we see from Lemma~\ref{lem:weitz} that a sufficient condition for the
absence of zeros of $Z_G(\beta)$ is the existence of a subset
$D\subseteq \C$ such that $1\in D$, $-1\notin D$, and $D$ is closed
under the recurrence~$F_{\beta,k, s}$ (in the sense that 
$F_{\beta,k, s}$ maps $D^k$ into~$D$).  These properties guarantee
that the recurrence, with initial inputs~1 at the leaves, can never
yield the value~$-1$, and hence that $R_{G, v}(\beta) \neq -1$, so
$Z_G(\beta)\ne 0$.  The main technical content of this paper is to 
prove, under the conditions on~$\beta$ stated in Theorem~\ref{thm:main-fisher},
the existence of such a set~$D$, a result which we formally state as follows.
\begin{theorem}
  Fix a degree $\Delta = d + 1$. For any
  $\beta \in \bigl(\frac{\Delta - 2}{\Delta}, \frac{\Delta}{\Delta -
      2}\bigr)$, there exists $\delta_\beta > 0$ such that, for any
  $\beta' \in \C$ with $\abs{\beta' - \beta}\leq\delta_\beta$, there
  exists a set $D \subseteq \C$ with $1 \in D$, $-1 \not\in D$, and
  \begin{enumerate}[(a)]
  \item $F_{\beta',k, s}(D^k) \subseteq D$ for integers $k \geq 0$ and
    $s$ such that $1 \leq k + \abs{s} \leq d$;
  \item $-1 \notin F_{\beta', k, s}(D^k)$ for integers $k \geq 0$ and
    $s$ such that $1 \leq k + \abs{s} \leq \Delta$.
  \end{enumerate}
  \label{lem:existence-of-D}
\end{theorem}
\noindent
At the end of this section, we spell out the details of how to combine Lemma~\ref{lem:weitz}
and Theorem~\ref{lem:existence-of-D} into a proof of our main result, Theorem~\ref{thm:main-fisher}.

The rest of the paper focuses on proving Theorem~\ref {lem:existence-of-D}.  We briefly
sketch our approach here.  The first step is to simplify the problem by working with a 
{\it univariate\/} version of the recurrence~$F_{\beta,k, s}$ defined in Lemma~\ref{lem:weitz}.
The univariate version is defined as $f_{\beta,k,s}(x) := \beta^s h_\beta(x)^k$, and we
can show that it satisfies $F_\beta(D^k) = f_\beta(D)$ for any set $D$ such that
$C:=\log(h_\beta(D))$ is convex in the complex plane.  (Henceforth we will drop the
subscripts $k,s$ for simplicity.)  This means that the set~$D$
we seek in Theorem~\ref {lem:existence-of-D} should be the image of a convex set~$C$
under the map $\log{}\circ h_\beta$.

Next, to enable us to exploit the fact that $\beta$ is in the correlation decay interval
$B=\bigl(\frac{\Delta - 2}{\Delta}, \frac{\Delta}{\Delta - 2}\bigr)$, we further modify 
the univariate recurrence to $f_\beta^\varphi := \varphi\circ f_\beta\circ \varphi^{-1}$, where
$\varphi(x):=\log x$.  This is an example of the use of a so-called ``potential''
function~$\varphi$ in order to smooth a recurrence, as has been useful in
several correlation decay arguments.  The key point here is that, when 
$\beta\in B$, $f_\beta^\varphi$ (unlike $f_\beta$ itself)
is actually a uniform {\it contraction\/} on an appropriate
domain in~$\C$; hence we can conclude that $f_\beta^\varphi(S)\subseteq S$ for
``nice'' sets~$S$ (i.e., $S$ that are convex and symmetric around the origin).
Since the condition $f_\beta(D)\subseteq D$ is equivalent to 
$f_\beta^\varphi(\log D)\subseteq\log D$,
this imposes the further constraint that $\log D$ be a ``nice'' set.

Putting together the constraints in the previous two paragraphs, we
need to construct a suitable convex set~$C$ whose image
$\log(h_\beta^{-1}(\exp(C)))$ is nice; our set $D$ in
Theorem~\ref{lem:existence-of-D} will then be defined as
$h_\beta^{-1}(\exp(C))$ (and this set must include~1 and
exclude~$-1$).  This turns out to be hard to achieve directly due to
the complexity of the map $p:=\log{}\circ h_\beta^{-1} \circ\exp$.
However, we are able to show that one can instead work with a
(non-analytic) approximation of~$p$ under which the image of a natural
convex~$C$ becomes a nice (in fact, rectangular) set.  Moreover, this
holds even for complex~$\beta$ that are sufficiently close to the
region~$B$.  This fact then allows us to push through the analysis and
arrive at a proof of Theorem~\ref{lem:existence-of-D}. Groundwork for
implementing the above strategy is laid in
Section~\ref{sec:preliminaries}.  The different approximations
required by the strategy outlined above lead to various geometric
considerations that are dealt with in
Section~\ref{sec:appr-funct-sets}.  The detailed proof of the theorem
then appears in Section~\ref{sec:finalproof}.

We conclude this overview section with the proofs of Lemma~\ref{lem:weitz}
and Theorem~\ref{thm:main-fisher} promised earlier.  The remainder of the paper will then be devoted
to proving our main technical result, Theorem~\ref{lem:existence-of-D}.

\begin{proof}[Proof of Lemma~\ref{lem:weitz} (Sketch)]
  Except for a few minor differences, this description is exactly the
  same as the version of Weitz's result used for the Ising model in,
  e.g., Refs.~\citenum{zhaliabai09} and
  \citenum{sinclair_approximation_2012}; for completeness, we describe
  the version of Weitz's SAW tree construction used in these
  references in Appendix~\ref{sec:weitzs-self-avoding}.  In the above
  references, only the maps $F_{\beta, k, 0}$ for $1 \leq k \leq d$
  (with at most one final application with $k = \Delta$, at the root
  of the SAW tree) are used, and the initial values come from the set
  $\set{0, \infty}$; these initial
  values are the values of the ratio for single leaf vertices in the
  SAW tree pinned to~$-$  and $+$ respectively, and the maps
  $F_{\beta, k, 0}$ describe how to combine the ratios from $k$ subtrees.
  The version in the lemma follows by noticing that
  $h_\beta(1) = 1, h_\beta(0) = \beta$ and
  $h_\beta(\infty) = 1/\beta$, so that $F_{\beta, k, 0}$ applied to a
  vector $\vec{x}$ with $k$ coordinates, $s_1$ of which are set to $0$
  and $s_2$ to $\infty$, produces the same output as
  $F_{\beta, k - \inp{s_1 + s_2}, s_1 - s_2}$ applied to the vector
  $\vec{x'}$ of $k - \inp{s_1 + s_2}$ coordinates obtained from
  $\vec{x}$ by removing the $0$ and $\infty$ entries. %
\end{proof}

\begin{proof}[Proof of Theorem~\ref{thm:main-fisher}]
As indicated earlier, the induction is on the number of unpinned
vertices, $n$, of~$G$.
For the base case $n=0$, $Z_G(\beta) = \beta^k$, where $k$ is the number of
pairs of adjacent vertices in $G$ that are pinned to different spins.
Therefore, $Z_G(\beta) \neq 0$ unless $\beta=0$.  Next suppose that for
some positive integer $t$, it holds that for every $\beta \in B$,
there exists a $\delta>0$ such that for all $\beta' \in \C$ with
$\abs{\beta' - \beta} < \delta$, $Z_G(\beta') \neq 0$ for all graphs
$G$ of maximum degree $\Delta$ with at most $t$ unpinned vertices.
Now, let $G'$ be any graph of the same maximum degree with $t+1$
unpinned vertices.  Fix any non-isolated vertex $v$ in $G'$, and let
$Z_{G', v}^+(\beta'), Z_{G', v}^-(\beta')$ be the contributions to the
partition function from configurations with $\sigma(v)=+,\sigma(v)=-$,
respectively.  By the induction hypothesis, we know that
$Z_{G', v}^+(\beta') \neq 0, Z_{G', v}^-(\beta') \neq 0$ as they are exactly
the Ising partition function defined on the same graph $G'$ with the
vertex $v$ pinned (thus reducing the number of unpinned vertices to
$t$).  Further, Lemma\nobreakspace \ref {lem:weitz} implies that
$R_{G', v}(\beta') = \frac{Z_{G', v}^+(\beta')}{Z_{G', v}^-(\beta')}$ can be
computed by iteratively applying a sequence of maps of the form
$F_{\beta', k, s}$ for $1 \leq k + \abs{s} \leq d$, followed by at
most one application where $k + \abs{s} = \Delta$, starting with
initial values of $1$.  Part~(a) of Theorem\nobreakspace \ref
{lem:existence-of-D} then implies that the outputs of all except
possibly the final application remain in the set $D$ defined in that
theorem, and part (b) of the theorem implies that the final output,
which is equal to $R_{G', v}(\beta)$ by Lemma~\ref{lem:weitz}, is
not~$-1$.  Since $Z_{G', v}^+(\beta')$ and $Z_{G', v}^-(\beta')$ are
non-zero, this implies that $Z_{G'}(\beta') \neq 0$, completing the
induction.
\end{proof}

\section{Preliminaries}
\label{sec:preliminaries}
\subsection{Tree recurrence and correlation decay}
We consider first the following univariate version of the recurrence
$F_{\beta, k, s}$ defined above:
\begin{equation}
  \label{eq:13}
  f_{\beta, k, s}(x) \defeq \beta^sh_\beta(x)^k,
\end{equation}
where as before $h_\beta(x) \defeq \frac{\beta +x}{\beta x + 1}$.
Both the multivariate and univariate recurrences have been
studied extensively in the literature on the Ising model on trees.  It has
also been found useful to re-parameterize the recurrence in terms of
logarithms of likelihood ratios as follows (see, e.g.,
Ref.~\citenum{lyons_ising_1989}).  Let $\varphi(x) \defeq \log x$ and define
\begin{equation}
  \label{eq:2}
  \fpsk \defeq \varphi \circ f_{\beta,k, s} \circ
  \varphi^{-1} = s \log \beta + k \log h_\beta(e^x).
\end{equation}
One then has the following ``step-wise'' version of correlation decay.\cite{lyons_ising_1989,zhaliabai09}

\begin{proposition} 
  Fix a degree $\Delta = d + 1$ and integers $k \geq 0$ and $s$. If
  $\frac{\Delta - 2}{\Delta} < \beta <\frac{\Delta }{\Delta - 2}$ then
  there exists an $\epsilon > 0$ (depending upon $\beta$ and $d$) such
  that $|{\fpsk}'(x)| < \frac{k}{d}(1 - \epsilon)$ for every $x \in \R$.
  \label{lem:correlation-decay}
\end{proposition}
\begin{proof}
  By direct calculation, one has
  \begin{displaymath}
    \bigl|{\fpsk}'(x)\bigr| = \frac{k \abs{1-\beta^2}}{\beta^2 +
      1 + \beta(e^x + e^{-x})}.
  \end{displaymath}
  By the AM-GM inequality, $e^x + e^{-x} \ge 2$ for every real $x$,
  and the left hand side is therefore at most
  $\frac{k}{d} \times d \times \frac{\abs{1-\beta}}{1+\beta}$ which is
  strictly smaller than $\frac{k}{d}$ under the condition on $\beta$.
\end{proof}

For any integers $k \geq 0$ and $s$ %
and a positive real $\beta$, we have
$\beta^{k + \abs{s}} \le f_{\beta,k,s}(x) \le \frac{1}{\beta^{k +
    \abs{s}}}$ when $\beta \leq 1$, and
    $ \frac{1}{\beta^{k + \abs{s}}}\le f_{\beta,k,s}(x) \le \beta^{k + \abs{s}}$ for $\beta\ge 1$.
Taking the logarithm of these bounds motivates the definition of the
intervals $I_0(\beta,k)$ as follows:
\begin{equation}
  I_0 = I_0(\beta, k) \defeq \inb{-2k\abs{\log\beta}, 2k\abs{\log\beta}}.
\label{eq:3}
\end{equation}

We now note some consequences of the contraction shown in
\MakeUppercase Proposition\nobreakspace \ref {lem:correlation-decay}
for the behaviour of the recurrence on the complex plane.  We write
$f^\varphi_{\beta, k}$ in place of $f^{\varphi}_{\beta, k, 0}$ to simplify
notation.
\begin{lemma}
  \label{lem:strict-contraction}
  Fix a degree $\Delta = d + 1$, and let $\beta \in B = \betainv$.
  Then there exist positive real constants
  $\delta_\beta, \epsilon, \eta$ and $M$ such that the following is
  true.  Let $C_0 = C_0(\beta, d)$ be the set consisting of all points
  within distance $\epsilon$ of $I_0(\beta, d)$ in $\C$, and let $B_0$
  be the set of all $\beta' \in \C$ such that
  $\abs{\beta' - \beta} < \delta_\beta$.  Then, for every $x \in C_0$,
  $\beta' \in B_0$ and a positive integer $k \leq d$, the function
  $\fpk$ and its derivative ${\fpk}' = \D{\fpk}{x}$ are both
  $M$-Lipschitz as functions of $\beta$ in an open neighborhood of
  $\beta$, and ${\fpk}'$ further satisfies
  \begin{equation}
    \sup_{x_0 \in C_0, \beta' \in B_0}
    \bigl|{\fpk}'(x_0)\bigr| < \frac{k}{d}(1- \eta).
    \label{eq:4}
  \end{equation}
  Moreover,
  \begin{enumerate}[(a)]
  \item $\fpk(x)$ preserves the real axis: for any
    $x\in\R,\beta\in\R$, we have $\fpk(x)\in \R$;\label{item:1}
  \item $\fpk(x)$ preserves the imaginary axis: for any purely
    imaginary number $x$ and $\beta\in \R$, $\fpk(x)$ is also purely
    imaginary;\label{item:2}
  \item for any $x'\in C_0$ and $\beta' \in B_0$,
    $\bigl|\Im\bigl( f_{\beta',k}^\varphi(x') \bigr)\bigr| \le
    \frac{k}{d}(1-\eta) \abs{\Im(x')} + M \abs{\beta' -
      \beta}$;\label{item:3}
  \item for any $x'\in C_0$ and $\beta' \in B_0$,
    $\bigl|\Re\bigl( f_{\beta',k}^\varphi(x') \bigr)\bigr| \le
    \frac{k}{d}(1-\eta) \abs{\Re(x')} + M \abs{\beta' -
      \beta}$.\label{item:4}
  \end{enumerate}
\end{lemma}
\begin{proof}
  We know from Proposition\nobreakspace \ref {lem:correlation-decay}
  that there exists $\eta > 0$ (depending only upon $\beta$ and $d$)
  such that for every $x \in I_0(\beta, d)$,
  $\bigl|{\fpk}'(x)\bigr| \leq \frac{k}{d}(1 - 2\eta)$.  Note also that
  $\fpk$ is analytic in both $\beta$ and $x$ when $\beta$ and $x$ lie
  in suitably small open neighborhoods (in $\C$) of $\beta$ and $I_0$
  respectively.  The existence of suitable constants
  $\delta_\beta, \epsilon$ and $M$, and eq.\nobreakspace \textup
  {(\ref {eq:4})} then follows from the compactness of the interval
  $I_0$. 

  Now for part~(\ref{item:1}), when both $x\in\R,\beta\in\R$, it follows
  from the definition of $\fpk = \log \circ f_{\beta,k} \circ \exp$ that
  $\fpk(x) \in \R$.  For part~(\ref{item:2}), we consider a purely imaginary
  number $y\iota$ and $\beta \in \R$ and use again the preceding
  definition of $\fpk$.  %
  Note first that $\exp(yi)$ lies on the unit circle, and hence is
  mapped again to the unit circle by the M\"obius transform
  $\frac{\beta+x}{\beta x+1}$ when $\beta \in
  \R$. $f_{\beta,k}(\exp(y\iota))$ also therefore lies on the unit circle,
  and its logarithm is therefore a purely imaginary number as required.

  For part~(\ref{item:3}), consider the real number $x=\Re(x')$, and note
  that $\fpk(x)$ is also a real number. Then,
  \begin{align*}
    \bigl|\Im\bigl(f_{\beta',k}^\varphi(x') \bigr)\bigr|
    \le &\bigl| f_{\beta',k}^\varphi (x') - f_{\beta,k}^\varphi (x) \bigr| %
          \le \bigl|f_{\beta',k}^\varphi(x') - f_{\beta',k}^\varphi(x)\bigr|  +  \bigl| f_{\beta',k}^\varphi(x) - f_{\beta,k}^\varphi(x)\bigr| \\
    \le & \abs{\int_\gamma {f_{\beta',k}^\varphi}' (z) dz } + M\abs{  \beta' - \beta},  \hbox{where $\gamma$ is the line segment $l(x, x')$}\\
    \le &\abs{\int_\gamma \abs{{f_{\beta',k}^\varphi}' (z)} dz } + M \abs{  \beta' - \beta}\\
    \le &\abs{x'  - x}\sup_{z \in C_0}\bigl|{f_{\beta',k}^\varphi}' (z)\bigr| + M\abs{  \beta' - \beta} \\
          \le &\frac{k}{d}(1-\eta) \abs{\Im(x')} + M\abs{  \beta' - \beta}.
  \end{align*}

  For part~(\ref{item:4}), consider the purely imaginary number
  $x = \Im(x') \iota$, so that thus $\fpk(x)$ is also a purely
  imaginary number. Then, as in the case of part~(\ref{item:3}),
  \begin{align*}
    \bigl|\Re\bigl( f_{\beta',k}^\varphi(x') \bigr)\bigr|
    \le & \bigl| f_{\beta',k}^\varphi (x') - \fpk (x)\bigr| %
          \le \bigl| f_{\beta',k}^\varphi(x') - f_{\beta',k}^\varphi(x) \bigr|  +  \bigl| f_{\beta',k}^\varphi(x) - \fpk(x)\bigr| \\
    \le &\abs{x' - x} \sup_{z \in C_0}\bigl|{f_{\beta',k}^\varphi}'(z)\bigr| + M\abs{  \beta' - \beta}\\
    \le & \frac{k}{d}(1-\eta) \abs{\Re(x')} + M\abs{  \beta' - \beta}.\qedhere
  \end{align*}
\end{proof}

\subsection{Reduction to the univariate case} We now show that, as in
the analysis of several correlation decay algorithms (see, e.g. the
arguments in
Refs.~\citenum{sinclair_approximation_2012,li_correlation_2011}, and
more recently, Ref.~\citenum{peters17:_sokal}),
we can restrict our attention to the
univariate version of the recurrence by exploiting a suitable notion of convexity.
Recall that
\begin{equation}
  \label{eq:5}
  h_\beta(x) \defeq  \frac{\beta+ x }{\beta x + 1},
\end{equation}
so that the univariate recurrence
$f_{\beta, k, s}(x) = \beta^s h_\beta(x)^k = \beta^sf_{\beta, k}(x)$.

\begin{lemma}
  Fix an arbitrary complex $\beta$.  For any set $D$ such that
  $\log\left( h_\beta(D) \right)$ is convex in the complex plane, we
  have $F_{\beta,k, s}(D^k) = f_{\beta, k, s}(D)$ for all integers
  $k \geq 0$ and $s$.
  \label{lem:convexity-reduction}
\end{lemma}
\begin{proof}
  The inclusion $f_{\beta, k, s}(D) \subseteq F_{\beta, k, s}(D^k)$ is
  trivial.  For the converse, consider $\vec{x} \in D^k$.  We then have
  \begin{align*}
    \log F_{\beta,k, s}(\vec{x}) = s\log \beta + \sum_{i=1}^k \log h_\beta(x_i).
  \end{align*}
  Now, by the convexity of $\log \left( h_\beta(D) \right)$, there
  exists $\tilde{x} \in D$ such that
  $s\log\beta + \sum_{i=1}^k \log h_\beta(x_i) = s \log \beta + k \log
  h_\beta(\tilde{x})$, which in turn equals
  $ \log f_{\beta, k, s}(\tilde{x})$.  It follows that
  $F_{\beta, k, s}(D^k) \subseteq f_{\beta, k, s}(D)$.
\end{proof}

Our strategy now will be to start with an appropriate convex set $C_2$
in the complex plane, and for a real $\beta \in B$ and a complex
$\beta'$ close to $\beta$, {\it define\/} $D$ as the region
$h_{\beta'}^{-1}(\exp(C_2))$.  Here $C_2$ will have to be so chosen
that its preimage~$D$ is closed under application of the univariate
recurrence (i.e., $f_{\beta', k, s}(D) \subseteq D$
for all integers $k \geq 0$ and $s$ such that
$1 \leq k + \abs{s} \leq \Delta - 1$).

To establish this latter claim, we will consider instead the action of
the maps $f_{\beta',k, s}^\varphi$ on the set
$C_1 \defeq \log (D) = (\log {\circ} h_{\beta'}^{-1} {\circ}
\exp)(C_2)$, and show that $f_{\beta', k, s}^\varphi(C_1) \subseteq C_1$
(which is equivalent to $f_{\beta', k, s}(D) \subseteq D$).  We know
from \MakeUppercase Lemma\nobreakspace \ref {lem:strict-contraction}
that $f_{\beta', k, 0}^\varphi$ is a contraction in an appropriate
domain.  Note that, on sets $S$ that are convex and symmetric around
the origin (which we call ``nice''), this contraction property implies
that $f_{\beta', k, 0}^\varphi$ maps $S$ into itself.  (Extending this to
$f_{\beta', k, s}$ with non-zero $s$ requires a careful argument,
which we defer to the full proof in Section~\ref{sec:finalproof}.)

Unfortunately, we are unable to prove directly that $C_1$ is ``nice''
due to the complexity of the map
$p_{\beta'} \defeq \log \circ h_{\beta'}^{-1} \circ \exp$ taking $C_2$
to $C_1$.  Instead we will work with a (non-analytic) approximation
$q$ of $p_\beta$, and further use the fact that $p_\beta$ is an
approximation of $p_{\beta'}$.  We will show that $q(C_2)$
 is a ``nice'' (in fact, rectangular) region, and then show in the following
section how such approximations can be ``chained'' to establish that
$C_1$ is indeed mapped into itself by $f_{\beta',k, s}^\varphi$.

\subsection{Rectangular sets and contraction}
Given non-negative reals $a$ and $b$, we define
\[
  R(a,b)\defeq \set{x+y \iota : \abs{x} \le a \text{ and } \abs{y} \le
    b}.
\]

We will need to consider maps which contract all but very small
rectangles in a given set.  For later use, we define this notion not
just for functions $\C \rightarrow \C$ but also for maps
$2^{\C} \rightarrow 2^{\C}$ which map subsets of complex numbers to
subsets.
\begin{definition}[\textbf{$(\chi, \tau, \xi)$-contraction of
    rectangles}] A map $g: 2^{\C} \rightarrow 2^{\C}$ is said to
  \emph{$(\chi, \tau, \xi)$-contract rectangles} in a given subset $C$
  of the complex plane if for every rectangle $R({a}, {b})$ contained
  in $C$ and $z \in g(R({a}, {b}))$, we have
  \begin{align*}
    \abs{\Re (z)} &\le \chi \max\set{a, \tau}; \\
    \abs{\Im (z)} &\le \chi \max\set{b, \xi}.
  \end{align*}
\end{definition}

We will now establish two facts which underlie the use of
rectangular sets for our purposes.  The first of these is the simple
observation that maps which contract rectangles either map the set
into itself or contract it . (Again for later use, we prove this for the
more general case of maps from $2^{\C}$ to $2^{\C}$.)  The second is
the fact that the maps $f_{\beta, k}^\varphi$ we considered above do
contract rectangles.

\begin{lemma}
  Let $C = R(\alpha_1, \alpha_2)$ be a rectangular set.  Let
  $\chi <1, \tau \leq \alpha_1$ and $\xi \leq \alpha_2$ be positive
  constants.  Let $g: 2^{\C} \rightarrow 2^{\C}$ be a map which
  $(\chi, \tau, \xi)$-contracts rectangles in $C$. Then
  $g\left(C\right) \subseteq \chi C$.
  \label{lem:closed-ideal} 
\end{lemma}
\begin{proof}
	Since $g$ is a map that  $(\chi, \tau, \xi)$-contracts rectangles
    in $C$, and $C$ itself is a rectangular set, we have, for any $z
    \in g(C)$,
	\begin{align*}
		\abs{\Re(z)} \le \chi \max\set{\alpha_1,\tau} \le \chi \alpha_1; \\
		\abs{\Im(z)} \le \chi \max\set{\alpha_2,\xi} \le \chi \alpha_2.
	\end{align*}
	Thus $g(C) \subseteq \chi C$.
\end{proof}

We now show that the $2^{\C} \rightarrow 2^{\C}$ map naturally induced
by $\fpk = f^\varphi_{\beta, k, 0}$ satisfies the hypothesis of the above lemma for rectangular
sets contained in $C_0$.

\begin{lemma}
  \label{lem:rect-contraction}
  Fix a degree $\Delta = d + 1$, $\beta \in \betainv$, and let
  $\eta, M$, $B_0 = B_0(\beta,d)$ and $C_0 = C_0(\beta,d)$ be as defined
  in \MakeUppercase Lemma\nobreakspace \ref{lem:strict-contraction}.
  Then for any positive constants $\tau$ and $\xi$, positive integer
  $k$ such that $1 \leq k \leq d$, and any $\beta' \in B_0$ such that
  $\abs{\beta' -\beta} \leq \eta\min\set{\xi,\tau}/(2Md)$, we have that
  $f_{\beta', k}^\varphi$ is a map that
  $(\frac{k}{d}(1 - \frac{\eta}{2}), \tau, \xi)$-contracts rectangles
  in $C_0$.
\end{lemma}
\begin{proof}
  It suffices to show that, for any
  $z \in C_0$, we have
  \begin{align*}
    \bigl| \Re\bigl(\fpb(z)\bigr)\bigr| \le {\textstyle\frac{k}{d}}(1-\eta/2)\max\set{\abs{\Re\left( z \right)}, \tau};\\
    \bigl| \Im\bigl(\fpb(z)\bigr)\bigr| \le {\textstyle\frac{k}{d}}(1-\eta/2)\max\set{\abs{\Im\left( z \right)}, \xi}.
  \end{align*}
  Now, since $\beta' \in B_0$, we have from part~(\ref {item:3}) of
  Lemma\nobreakspace \ref {lem:strict-contraction} that
  \begin{align*}
    \bigl|\Im\bigl( \fpb(z) \bigr)\bigr|
    \le {\textstyle\frac{k}{d}}(1-\eta) \abs{\Im(z)} + M\abs{\beta' - \beta}
    \le {\textstyle\frac{k}{d}}(1-\eta) \abs{\Im(z)} + \eta \xi/(2d) .
  \end{align*}
  If $\abs{\Im(z)} \geq \xi$, then
  $\frac{k}{d}(1-\eta) \abs{\Im(z)} + \eta \xi/(2d) \le \frac{k}{d}(1-\eta) \abs{\Im(z)} +
  \eta/(2d) \abs{\Im(z)} \leq \frac{k}{d}(1-\eta/2) \abs{\Im(z)}$.  Otherwise if
  $\abs{\Im(z)} < \xi$, then
  $\frac{k}{d}(1-\eta) \abs{\Im(z)} + \eta \xi/(2d) \le \frac{k}{d}(1-\eta) \xi + \eta \xi/(2d) \leq
  \frac{k}{d}(1-\eta/2) \xi$.  Therefore we have
  \[
    \bigl|\Im\bigl(\fpb(z)\bigr)\bigr| \le
    {\textstyle\frac{k}{d}}(1-\eta/2)\max\set{\abs{\Im\left( z \right)}, \xi}.
  \]
  A symmetrical argument using part~(\ref{item:4}) of Lemma\nobreakspace \ref {lem:strict-contraction}
  gives the desired upper bound on $\bigl|\Re\bigl( \fpb(z) \bigr)\bigr|$.
\end{proof}

We can now describe the convex set $C_2$.  Given a degree
$\Delta = d +1$, $\beta > 0$, and $\delta > 0$, we define for
$0 \leq k \leq d$
\begin{equation}
  \begin{aligned}
    C_2(\beta,\delta,k) &\defeq \set{z: \abs{\Re(z)} \le \bigl|\log
        h_\beta(\beta^{2k})\bigr|, \abs{\Im(z)} \le i_{k, \delta}(\Re(z))},
    \text{ where }\\
    i_{k, \delta}(x) &\defeq k\delta( \cosh (\log \beta) - \cosh
    (x)) \frac{2\beta}{\abs{1-\beta^2}}.
  \end{aligned}
  \label{eq:7}
\end{equation}
Note that when $k \geq 1$, $i_{k, \delta}$ is an even function that is
continuous, decreasing and concave for positive $x$.  Further, when
$k \geq 1$, we also have
$i_{k, \delta}\inp{\abs{\log h_\beta(\beta^{2k})}} > 0$ for all
positive $\beta \neq 1$ (as can be verified by checking that
$\abs{\log h_\beta(\beta^{2k})} = \log (\beta^{2k + 1} + 1) - \log (\beta +
\beta^{2k})$, and that this in turn is always smaller than $\abs{\log \beta}$
when $\beta \neq 1$).
In particular, these facts imply that $C_2$ is convex.
\begin{observation}
  For $\beta, \delta > 0$, $\beta \neq 1$ and $k_1 \geq k_2 \geq 0$,
  $C_2(\beta, \delta, k_1) \supseteq C_2(\beta, \delta, k_2)$.\label{obv:ascend}
\end{observation}
\begin{proof}
  This follows from the fact that $\abs{\log h_\beta(\beta^{2k})}$ is
  monotonically increasing in $k$ when $\beta \neq 1$ and $k \geq 0$.  The
  latter fact is a consequence of the monotonicity of $h_\beta$ on the
  positive real line (it is monotonically increasing when $\beta < 1$ and
  monotonically decreasing when $\beta > 1$), and the fact that
  $h_\beta(1/x) = 1/h_\beta(x)$.
\end{proof}

\section{Approximations for functions and sets}
\label{sec:appr-funct-sets}
As hinted in the discussion in the previous section, in order to apply
Lemmas\nobreakspace \ref {lem:closed-ideal} and\nobreakspace \ref
{lem:rect-contraction} directly we would need
$C_1 \defeq p_{\beta'}(C_2)$ to be a rectangular set, where for a
complex $\beta'$ close to $\beta$,
\begin{equation}
  \label{eq:6}
  p_{\beta'} \defeq \log \circ h_{\beta'}^{-1} \circ \exp.
\end{equation}
Unfortunately, we are not able to work directly with the somewhat
complicated map $p_{\beta'}$ in order to prove that $C_1$ is rectangular.

Instead we take the different approach of ``approximating'' $C_1$ by a
region we can directly prove to be rectangular.  We will do this
approximation in two steps.  We first consider
$\tilde{C} = p_\beta(C_2)$ as an approximation of $C_1$ and then
$q_\beta(C_2)$, where $q_\beta$ is a function approximating $p_\beta$,
as an approximation to $\tilde{C}$.  We then show that $q_\beta(C_2)$ is
indeed a rectangular set, and further, that the two approximations
above are good enough to allow one to show that $C_1$ itself is
mapped into itself by $f_{\beta'}^\varphi$.  The machinery for justifying
the approximations will be developed later in this section; first, we
describe $q_\beta$ and show that $q_\beta(C_2)$ is indeed rectangular.

Unlike $p_\beta$, which is analytic in a suitable neighborhood close
to the positive real line, $q_\beta$ will not be an analytic function.
Given $a + b\iota$ with $a, b \in \R$, we define
\begin{equation}
  q_\beta(a+b\iota) \defeq p_\beta(a) + p_\beta'(a)b\iota
  = \log \circ (h_\beta^{-1}) \circ \exp(a) +
  \frac{(1-\beta^2)}{2\beta(\cosh (\log \beta )- \cosh a)}b\iota.\label{eq:9}
\end{equation}
We will show that $q_\beta(C_2(\beta, \delta, k))$ is in fact a
rectangular region.  In preparation for this, we need the following facts:

\begin{observation}
The maps $p_\beta$ and $q_\beta$ have the following properties:
  \begin{enumerate}[(a)]
  \item $p_\beta'(a) = p_\beta'(-a)$.  Further, for $\beta > 0$ and $a$
    such that $\abs{a} < \abs{\log \beta}$, $p_\beta'(a)$ has the same
    sign as $1 - \beta^2$.\label{item:5}

  \item For $\beta > 0$, $\beta \neq 1$, and any positive integer $k$,
    $p_\beta$ maps the interval
    \[[-|\log h_\beta\inp{\beta^{2k}}|, |\log
      h_\beta\inp{\beta^{2k}}|\,]\] bijectively to the interval
    $[-2k\abs{\log \beta}, 2k\abs{\log \beta}\,]$.\label{item:6}

  \item For $\beta \neq 1$ and $\delta > 0$, $q_\beta$ bijectively
    maps the interior of $C_2(\beta, \delta, k)$ into the interior
    of $q_\beta(C_2(\beta, \delta, k))$, and the boundary
    $\partial C_2(\beta, \delta, k)$ to the boundary of
    $q_\beta(C_2(\beta, \delta, k))$. In particular,
    $\partial q_\beta(C_2(\beta, \delta, k)) = q_\beta(\partial
    C_2(\beta, \delta, k))$. \label{item:7}
  \end{enumerate}\label{obv:prop-q}
\end{observation}

\begin{proof}
\begin{enumerate}[(a)]
\item We have
      $p_\beta'(a) = \frac{(1-\beta^2)}{2\beta(\cosh (\log \beta )-
        \cosh a)}$, so that $p_\beta'(a) = p_\beta'(-a)$. Note that
      the denominator is always positive for $a$ such that
      $\abs{a} < \abs{\log \beta}$, since $\cosh x > \cosh y$
      whenever $\abs{x} > \abs{y}$.

    \item The previous part implies that $p_\beta$ is either
      monotonically strictly increasing (when $\beta < 1$) or
      monotonically strictly decreasing (when $\beta > 1$) in the open
      interval $(-\abs{\log\beta}, \abs{\log\beta})$.  As argued in
      the paragraph below eq.~\eqref{eq:7}, this interval contains the
      interval
      $[-|\log h_\beta\inp{\beta^{2k}}|, |\log
      h_\beta\inp{\beta^{2k}}|\,]$.  Thus, the fact that the map is
      bijective follows.  The values of $p_\beta$ at the endpoints are
      obtained by direct evaluation.

 \item This follows from the above two parts and the form of
      $q_\beta$, since the boundary of $C_2$ is given by the two
      vertical line segments
      $\bigl\{\pm |\log h_\beta(\beta^{2d})| + y\iota \st \abs{y}
        \leq i_{k, \delta}(|\log h_\beta(\beta^{2k})|)\bigr\}$ along with the
      curves
      $\bigl\{x \pm i_{k, \delta}(x)\iota \st \abs{x} \leq |\log
          h_\beta(\beta^{2k})|\bigr\}$. 
    \end{enumerate}
  \end{proof}
The property in part~(\ref{item:7}) of Observation\nobreakspace \ref {obv:prop-q} in fact holds also
for $p_\beta$ and $p_{\beta'}$ for $\beta'$ close to $\beta$, so we
make a note of this for future use.
\begin{observation}
  Let $d$ be a positive integer and $\beta$ a positive real such that
  $\beta \neq 1$.  There exists a $\delta' >0$ such that, for any
  $\beta'$ with $\abs{\beta' - \beta} \leq \delta'$ and
  $0 < \delta \leq \delta'$, $p_{\beta'}$ is analytic on
  $C_2(\beta, \delta, d)$ and also has an analytic inverse.  Further,
  $\partial p_{\beta'}(C_2(\beta, \delta, k)) = p_{\beta'}(\partial
  C_2(\beta, \delta, k))$ for any non-negative integer $k \leq
  d$. \label{obv:analytic-p}
\end{observation}

\begin{proof}
  Recall that
  $p_{\beta'} \defeq \log \circ h_{\beta'}^{-1} \circ \exp$. Note that
  for $\delta' > 0$ small enough and any $\beta'$ within distance
  $\delta'$ of $\beta$, the real part of any $z$ in
  $(h_{\beta'}^{-1} \circ \exp)(C_2(\beta, \delta', d))$ (where
  $0 < \delta \leq \delta'$) is strictly positive. Thus, the standard
  branch of the $\log$ function is analytic with an analytic inverse
  on the sets $(h_{\beta'}^{-1} \circ \exp)(C_2(\beta, \delta, d)))$.
  The analyticity of $p_{\beta'} $ and the existence of the analytic
  inverse $p_{\beta'}^{-1} = \log \circ h_{\beta'}^{-1} \circ \exp$
  thus follow for such $\beta'$ and $\delta$.

  The final claim is trivially true when $k = 0$ since in this case
  $C_2(\beta, \delta, k) = \inbr{0}$. We therefore assume $k \geq
  1$. Let $C' \defeq p_{\beta'}(C_2(\beta, \delta, k))$.  By the
  compactness and connectedness of $C_2(\beta, \delta, k)$ and the
  continuity of $p_{\beta'}$, $C'$ is also compact and connected. By
  the open mapping theorem applied to $p_{\beta'}$ and
  $\interior{C_2(\beta, \delta, k)}$, we see that
  $p_{\beta'}(\interior{C_2(\beta, \delta, k)}) \subseteq
  \interior{C'}$.  This implies that
  $\partial C' \subseteq p_{\beta'}(\partial C_2(\beta, \delta, k))$.
  On the other hand, an application of the open mapping theorem to
  $p_{\beta'}^{-1}$ and $\interior{C'}$ shows that
  $p_{\beta'}^{-1}(\interior{C'}) \subseteq \interior{C_2(\beta,
    \delta, k)}$, which implies that
  $p_{\beta'}(\partial C_2(\beta, \delta, k)) \cap \interior{C'} =
  \emptyset$.  It therefore follows that $\partial C'$ is in fact
  equal to $p_{\beta'}(\partial C_2(\beta, \delta, k))$.
\end{proof}

We can now show that $q_\beta(C_2(\beta, \delta, k))$ is indeed a
rectangular region.
\begin{lemma}
  \label{lem:image-q}  Let $k$ be any non-negative integer and $\delta$ and $\beta$ be positive
  reals such that $\beta \neq 1$.  Define $C_2(\beta, \delta, k)$ and
  $i_{k, \delta}(x)$ as in
  eq.\nobreakspace \textup {(\ref {eq:7})}.
  Then,
  \begin{equation*}
    q_\beta\left( C_2(\beta,\delta, k) \right) = \set{x: \abs{\Re(x)} \le
      2k\abs{\log \beta}, \abs{\Im(x)} \le k\delta} = k
    R(2\abs{\log \beta}, \delta).
    \label{eq:8}
  \end{equation*}
  Further, if positive integers $d$ and $\Delta = d + 1$ are such that
  $k \leq d$ and $\beta \in \betarange$, then there exists a positive
  $\delta_1 = \delta_1(\beta, \Delta)$ such that for all
  $\delta \leq \delta_1$, $q_\beta(C_2(\beta, \delta, k))$ is
  contained in the interior of the set $C_0(\beta, d)$ defined in
  \MakeUppercase Lemma\nobreakspace \ref {lem:strict-contraction}.
\end{lemma}

\begin{proof}
  The range of $\Re(x)$ follows from part~(\ref {item:6}) of
  Observation\nobreakspace \ref {obv:prop-q}.  To derive the range of
  $\Im(x)$, we first note that the maximum and minimum possible
  values of $b$ in eq.\nobreakspace \textup {(\ref {eq:9})} for a
  given value of $a > 0$ are given by $\pm i_{k, \delta}(a)$.  From part~(\ref {item:5})
  of Observation\nobreakspace \ref {obv:prop-q}, we then
  have that the set of all possible $\Im(q_\beta(a + \iota b))$ for all
  such $a$ and $b$ is precisely the interval
  $[-|{p_\beta'\inp{\abs{a}}}| \times i_{k, \delta}\inp{\abs{a}},
  |{p_\beta'\inp{\abs{a}}}| \times i_{k, \delta}\inp{\abs{a}}]$, which equals
  $[-k\delta, k\delta]$ using the form of $p_\beta'$ and
  $i_{k, \delta}$.

  The fact that $q_\beta(C_2(\beta, \delta,k))$ is contained in the
  interior of $C_0 = C_0(\beta, d)$ for small enough $\delta$
  follows from the facts that $i_{k, \delta}$ is proportional to
  $\delta$, that the real part of every point $z$ in
  $q_\beta(C_2(\beta, \delta,k))$ lies in the interval
  $I_0 = I_0(\beta, d)$ as defined in Lemma\nobreakspace \ref
  {lem:strict-contraction}, and that $C_0$ is an open neighborhood of
  $I_0$ in the complex plane. \qedhere
\end{proof}

We now show that $q_\beta$ is a meaningful approximation to $p_\beta$.
We begin with a simple observation.

\begin{lemma}
  \label{lem:image-p} Fix a positive $\beta \neq 1$ and a positive
  integer $d$.  There exist positive constants $M_1$ and $\delta_1$
  such that for all $\delta \leq \delta_1$, $0 \leq k \leq d$, and
  $C_2(\beta, \delta, k)$ as defined in eq.\nobreakspace \textup
  {(\ref {eq:7})},
  \[  \forall z \in C_2(\beta,\delta, k),\quad
  \abs{p_\beta(z)-q_\beta(z)} < M_1\delta^2/2. \]
\end{lemma}

\begin{proof}
  We have already argued in Observation\nobreakspace \ref
  {obv:analytic-p} that there exists $\delta' > 0$ (depending upon
  $\beta$ and $d$) such that $p_\beta$ is analytic in
  $C_2(\beta, \delta, d)$ for $\delta \leq \delta'$.  We therefore
  have a finite maximum, say $M_0$, for $|{p_\beta''(x)}|$ when $x$
  ranges over all $C_2(\beta, \delta, d)$ with $\delta \leq \delta'$.
  From a mean value argument applied to $p_\beta$, we then have, for
  any
  $a + b\iota \in C_2(\beta, \delta,k) \subseteq C_2(\beta, \delta,
  d)$,
  \[
    \abs{p_\beta(a+b \iota) - q_\beta(a+b \iota)} = \abs{p_\beta(a+b\iota) -
      p_\beta(a) - p_\beta'(a)b\iota} \le M_0b^2/2.
  \]
  Note that $\abs{b} \leq i_{k, \delta}(0) = \delta k
  \frac{\abs{\beta-1}}{\beta+1} \le \delta d \frac{\abs{\beta-1}}{\beta+1}$ for
  $a + b \iota  \in C_2(\beta, \delta,k)$, so that we can choose
  $M_1 > M_0 d^2 \inp{\frac{\abs{\beta-1}}{\beta+1}}^2 $ to get the
  claim. (Note that $M_1 > M_0$ suffices if $\beta$ is in the
  correlation decay interval $\betarange$ corresponding to
  $d = \Delta - 1$.)
\end{proof}

In a similar fashion, we can show that $p_{\beta}$ approximates $p_{\beta'}$ for
complex $\beta'$ close to $\beta$:
\begin{lemma}
  \label{lem:image-p2} Fix a positive $\beta \neq 1$ and a positive
  integer $d$.  There exist positive constants $\delta_1$ and $M_1$
  such that for all positive $\delta \leq \delta_1$, there exists a
  positive $\delta_2 > 0$ such that for all $\beta'$ with
  $\abs{\beta'-\beta} \leq \delta_2$ and $0 \leq k \leq d$, we have
  \[
    \forall z \in C_2(\beta,\delta, k), \quad |{p_\beta(z) -
      p_{\beta'}(z)}| < M_1\delta^2/2.
  \]
  Here, $C_2(\beta, \delta, k)$ is as defined in eq.\nobreakspace
  \textup {(\ref {eq:7})}.
\end{lemma}

\begin{proof}
  Arguments similar to those used in
  Observation\nobreakspace \ref {obv:analytic-p} show that for small
  enough positive $\delta'$ and~$\delta_1$, $p_{\beta'}(x)$ is analytic in
  both $\beta'$ and $z$ when $\abs{\beta' - \beta} \leq \delta'$ and
  $x \in C_2(\beta, \delta, d)$ for some $\delta \leq \delta_1$.  Hence we
  have a finite upper bound $M_0$ on
  $\frac{\partial}{\partial \beta'} p_{\beta'}(z)$ for such $\beta'$
  and $z$.  By a mean value argument, we then have for any such
  $\beta'$ and any $z \in C_2(\beta, \delta,k) \subseteq C_2(\beta,
  \delta, k)$,
  \[
    \abs{p_{\beta'}(z) - p_\beta(z)} \le M_0 \abs{\beta' - \beta}.
  \]
  We can then choose $M_1 > M_0$ and
  $\delta_2 \leq \min\inp{\delta', \delta^2/2}$ to get the
  claim.
\end{proof}

In order to study how these approximations act on subsets of the
complex plane, we use the
following notion of set approximation.

\begin{definition}[\textbf{Set approximation}]
  \label{def:set-approx} Let $\epsilon > 0$.
  Given a set $A$, we define
  \[
    \S(A) \defeq \bigcap_{\zeta\in\C: \abs{\zeta}\le \eps} \set{z +
      \zeta : z \in A},
  \]
  and
  \[
    \Sinv(A) = \bigcup_{\zeta\in\C: \abs{\zeta}\le \eps} \set{z + \zeta : z \in A}.
  \]
\end{definition}
Note that for any $\epsilon \geq 0$,
$\S(A) \subseteq A \subseteq \Sinv(A)$.  It is also easy to see that
for any collection $\mathcal{S}$ of sets,
$\S\inp{\bigcap_{A \in \mathcal{S}} A} = \bigcap_{A \in \mathcal{S}}
\S(A)$, and
$\Sinv\inp{\bigcup_{A \in \mathcal{S}} A} = \bigcup_{A \in
  \mathcal{S}} \Sinv(A)$. Further, $\Sinv$ is a ``pointwise
pseudoinverse'' of $\S$ in the sense that
$z \in \S(A) \iff \Sinv(\set{z}) \subseteq A$.  To see this, observe
that $z \in \S(A)$ iff  for all $\zeta$ such that
$\abs{\zeta}\le \eps$, $z - \zeta \in A$.  The latter is true iff
$\Sinv(\set{z}) \subseteq A$.  The latter fact implies that
the statements $\Sinv(A) \subseteq B$ and $A \subseteq \S(B)$ are
equivalent.  (Note, however, that the statements $A \subseteq \Sinv(B)$
and $\S(A) \subseteq B$ are \emph{not} in general equivalent.)  For
non-negative $\epsilon$ and $\delta$, we also have
\begin{align*}
  (S_\epsilon \circ S_\delta)(A) = S_{\epsilon + \delta}(A)\text{,
  and }
  (S_\epsilon^\dagger \circ S_\delta^\dagger)(A) = S_{\epsilon + \delta}^\dagger(A).
\end{align*}

This notion of set approximation can be related to 
function approximation via the following lemma.

\begin{lemma}
  Let $f$, $g$ be continuous maps on a compact subset $C$ of the
  complex plane such that
  $f(\partial C) = \partial f(C), g(\partial C) = \partial g(C)$.  If
  $g$ and $f$ are close in the sense that
  \[
    \forall z \in C, \abs{f(z) - g(z)} < \eps/2,
  \]
  then $\S(f(C)) \subseteq g(C) \subseteq \Sinv(f(C))$.
\label{lem:approx-image2}
\end{lemma}
\begin{proof}

  We first show that $g(C) \subseteq \Sinv(f(C))$.  Consider a point
  $g(z)$ for some $z \in C$, so $g(z) \in g(C)$.  Since $f$ and $g$
  are close, we see that $\zeta \defeq {f(z) - g(z)}$ has length less
  than $\eps/2$.  Therefore
  $g(z) \in \Sinv(\set{f(z)}) \subseteq \Sinv(f(C))$.

  Next, we show that $\S(f(C)) \subseteq g(C)$.  Note first that
  $\S(f(C)) \subseteq f(C)$, so that every point in the former is of
  the form $f(z)$ for some $z \in C$.  Now suppose, for the sake of contradiction,
  that for some $z \in C$ we have $f(z) \in \S(f(C))$ but
  $f(z) \not\in g(C)$.  Since $f$ and $g$ are close, we have
  $\abs{f(z) - g(z)} < \eps/2$.  Thus, if $f(z)$ is not in $g(C)$,
  there exists $y' \in \partial g(C)$ such that
  $\abs{f(z) - y'} < \eps/2$.  On the other hand, $f(z) \in \S(f(C))$
  means that $\Sinv(\set{f(z)}) \subseteq f(C)$, so that for all
  $y'' \in \partial f(C)$, $\abs{y'' - f(z)} \geq \eps$. We then have,
  for every $y'' \in \partial f(C)$,
  \begin{equation}
    \label{eq:10}
	\abs{y'' -y'}
	= \abs{y'' - f(z) + f(z) - y'}
	\ge \abs{ \abs{y'' - f(z)} - \abs{f(z) - y'} }
	> \eps/2.
  \end{equation}
We have thus shown that there exists $y' \in \partial g(C)$ such that
for all $y'' \in \partial f(C)$, $\abs{y'' - y'} > \eps/2$.  However,
since $y' \in \partial g(C)$ and $\partial g(C)$ is assumed to be
$g(\partial C)$, there must exist $z' \in \partial C$ such that
$g(z') = y'$.  Since $f$ and $g$ are close, we then have that
$\abs{f(z') - y'} < \epsilon/2$.  However, since we assumed that
$f(\partial C) = \partial f(C)$, we must also have
$f(z') \in \partial f(C)$ (since $z' \in \partial C)$.  But this is a
contradiction to eq.\nobreakspace \textup {(\ref {eq:10})}.
\end{proof}

The above lemma, along with our previous observations about the
properties of the maps $q_\beta$, $p_{\beta'}$ and $p_\beta$ now
allows us to lift our function approximations from 
Lemmas\nobreakspace \ref{lem:image-p} and\nobreakspace \ref {lem:image-p2}
to set approximations.

\newcommand{\St}{\ensuremath{S_{2\epsilon}}}
\newcommand{\Stinv}{\ensuremath{S_{2\epsilon}^\dagger}}
\begin{corollary}
  \label{cor:set-contract}
  Fix a positive $\beta \neq 1$ and a positive integer $d$, and let
  $C_2(\beta, \delta, k)$ for an integer $k$ such that
  $0 \leq k \leq d$ be as defined in eq.\nobreakspace \textup {(\ref
    {eq:7})}.  There exist positive constants $\delta_1$ and $M_1$
  (depending only upon $\beta$ and $d$) such that for every
  $\delta \leq \delta_1$, there exists a positive $\delta_2$ such that
  for all $\beta'$ with $\abs{\beta'-\beta}<\delta_2$, and all integers $k$
  with $0 \leq k \leq d$, we have
  \begin{enumerate}[(a)]
  \item \label{item:8} $\S(C(k)) \subseteq \tilde{C}(k) \subseteq \Sinv(C(k))$;

  \item \label{item:9}
    $\S(\tilde{C}(k)) \subseteq C_1(k) \subseteq \Sinv(\tilde{C}(k))$,
  \end{enumerate}
  where $\epsilon \defeq M_1\delta^2$,
  $C(k) \defeq q_\beta(C_2(\beta, \delta, k))$,
  $\tilde{C}(k) \defeq p_\beta(C_2(\beta, \delta, k))$ and
  $C_1(k) \defeq p_{\beta'}(C_2(\beta, \delta, k))$.  In particular,
  \begin{displaymath}
    \St(C(k)) \subseteq C_1(k) \subseteq \Stinv(C(k)).
  \end{displaymath}
\end{corollary}

\begin{proof}
  We choose $\delta_1$ to be the smaller and $M_1$ to be the larger of
  the corresponding quantities guaranteed by Lemmas\nobreakspace \ref
  {lem:image-p} and\nobreakspace \ref {lem:image-p2}.  Given this
  choice of $\delta_1$ and any $\delta \leq \delta_1$, we then choose
  $\delta_2$ to be the corresponding quantity guaranteed by
  Lemma~\ref{lem:image-p2}.

  Now, for part~(\ref{item:8}), we apply Lemma\nobreakspace \ref
  {lem:approx-image2} to the set $C_2(\beta, \delta, k)$ with
  $f = q_\beta$, $g = p_\beta$ and $\epsilon = M_1\delta^2$, and note
  that the topological conditions required on $f$ and $g$ are
  satisfied due to part~(\ref {item:7}) of
  Observations\nobreakspace \ref {obv:prop-q} and
  \nobreakspace \ref {obv:analytic-p} respectively, while
  the required closeness of $f$ and $g$ is guaranteed by
  Lemma\nobreakspace \ref {lem:image-p}.

  Similarly, for part~(\ref{item:9}), we apply Lemma\nobreakspace \ref
  {lem:approx-image2} again to the set $C_2(\beta, \delta, k)$ with
  $f = p_\beta$, $g = p_{\beta'}$ and $\epsilon = M_1\delta^2$.  In
  this case the topological conditions required on $f$ and $g$ are
  satisfied due to Observation\nobreakspace \ref {obv:analytic-p},
  while the closeness of $f$ and $g$ is guaranteed by
  Lemma\nobreakspace \ref {lem:image-p2}.

  The final claim follows from combining parts~(\ref{item:8}) and~(\ref{item:9})
  and using the properties of the maps $\S$ and $\Sinv$ described in the paragraph
  following Definition\nobreakspace \ref {def:set-approx}.  In particular, we have
  $C_1(k) \subseteq \Sinv(\tilde{C}(k)) \subseteq (\Sinv \circ
  \Sinv)(C(k)) = \Stinv(C(k))$ and similarly,
  $C_1(k) \supseteq \S(\tilde{C}(k)) \supseteq (\S \circ \S)(C(k)) =
  \St(C(k))$.
\end{proof}

Finally, the following lemma shows that the notion of contraction of
rectangles is robust with respect to these set approximations.
\begin{lemma}
  \label{lem:chain-contraction}
  Let $C$ be a subset of the complex plane. Let the positive constants
  $\chi, \eta \in (0, 1)$, $\epsilon_0, \tau$ and $\xi$ and the map
  $g: 2^\C \rightarrow 2^\C$ be such that it
  $(\chi(1-\eta), \tau, \xi)$-contracts rectangles in
  $S^\dagger_{\epsilon_0}(C)$.  Then for all
  $\epsilon \leq \min(\epsilon_0, \chi\eta\tau/4, \chi\eta\xi/4)$, the
  $2^\C \rightarrow 2^\C$ map
  $g^{\Sinv} \defeq \Sinv \circ g \circ \Sinv$ is a map that
  $(\chi (1 - \eta/2), \tau, \xi)$-contracts rectangles in $C$.
  Further, if $g\inp{\bigcup_{A \in \mathcal{S}} A} = \bigcup_{A \in
  \mathcal{S}} g(A)$, then $g^{\Sinv}$ also satisfies
  \begin{equation}
    g^{\Sinv} \inp{\bigcup_{A \in \mathcal{S}} A} = \bigcup_{A \in
      \mathcal{S}} g^{\Sinv}(A)\label{eq:11}
  \end{equation}
  for all collections $\mathcal{S}$ of subsets of $C$.
\end{lemma}
\begin{proof}
  Consider a rectangle $R = R(a, b) \subseteq C$. We have
  $\Sinv(R) \subseteq R(a + \epsilon, b + \epsilon)$. Note that for
  $\epsilon \leq \epsilon_0$,
  $S_{\epsilon_0}^{\dagger}(C) \supseteq \Sinv(C)$.  Thus, since $g$ is a map that
  $(\chi(1-\eta), \tau, \xi)$-contracts rectangles in
  $S_{\epsilon_0}^{\dagger}(C)$ (and hence also in $\Sinv(C)$), we
  have for any $z \in (\Sinv \circ g \circ \Sinv)(R)$,
  \begin{align*}
    \abs{\Re(z)} &\leq \epsilon + \chi(1-\eta)\max(a + \epsilon, \tau)\text{,
                   and}\\
    \abs{\Im(z)} &\leq \epsilon + \chi(1-\eta)\max(b + \epsilon, \xi).
  \end{align*}
  These conditions imply
  \begin{align*}
    \abs{\Re(z)} &\leq \chi(1-\eta/2)\max(a, \tau)\text{,
                   and}\\
    \abs{\Im(z)} &\leq \chi(1-\eta/2)\max(b, \xi).
  \end{align*}
  provided $\epsilon \leq \chi\eta \min(\tau, \xi)/4$, as claimed.  Thus under
  this further condition on $\epsilon \leq \epsilon_0$, $g^{\Sinv}$ is a map that
  $(\chi(1-\eta/2), \tau, \xi)$-contracts rectangles in $C$.

  For the final claim we have, for any collection $\mathcal{S}$ of
  subsets of $C$,
  \begin{align*}
    (\Sinv \circ g \circ \Sinv)\inp{\bigcup_{A \in \mathcal{S}} A}
    &= (\Sinv \circ g)\inp{\bigcup_{A \in \mathcal{S}} \Sinv(A)}\\
    &= \Sinv \inp{\bigcup_{A \in \mathcal{S}} (g \circ \Sinv (A))}\\
    &= \bigcup_{A \in \mathcal{S}} \inp{\Sinv \circ g \circ
      \Sinv}(A). \qedhere
  \end{align*}
\end{proof}

We now have all the machinery in place to prove our main technical
result, Theorem\nobreakspace \ref {lem:existence-of-D}.

\section{Proof of Theorem\nobreakspace \ref {lem:existence-of-D}}\label{sec:finalproof}
We first concentrate on the case $\beta \neq 1$; the special case
$\beta=1$ is simpler and will be dealt with at the end of the proof.
Given $\beta \in \betainv$, we show that if we choose $\delta_\beta$
and $\delta$ small enough, then for all $\beta'$ such that
$\abs{\beta-\beta'} \leq \delta_\beta$,
$D = (h_{\beta'}^{-1} \circ \exp)(C_2(\beta, \delta, d))$ satisfies
all the desired properties (note that the condition $1 \in D$ is
satisfied by design, since $0 \in C_2(\beta, \delta, d)$).  Let the sets
$I_0 = I_0(\beta, d)$, $C_0 = C_0(\beta, d)$ and the neighborhood
$B_0$ of $\beta$, all depending on $d$ and $\beta$, be as in
\MakeUppercase Lemma\nobreakspace \ref {lem:strict-contraction}.

We deal first with part~(b) of the theorem, which can be handled via
simple continuity arguments.  By Lemma\nobreakspace \ref
{lem:convexity-reduction}, we have
$F_{\beta', k, s}(D^{d+1}) = f_{\beta',k, s}(D)$, so it suffices to
show that $-1 \not\in f_{\beta',k, s}(D)$.  Note that there exists a
fixed positive constant $M$ (which can be taken to be
$e^{4d\abs{\log\beta}}$) such that, given any $\varepsilon > 0$, we
can choose $\delta$ and $\delta_\beta$ to be sufficiently small so
that for any $\beta'$ such that
$\abs{\beta - \beta'} \leq \delta_\beta$ and for any $z$ in the
corresponding $D$, $\abs{\Im(z)} \leq \epsilon$ and
$\frac{1}{M} \leq \Re(z) \leq M$ (note that this already establishes
that $-1 \not\in D$ when $\delta$ and $\delta_\beta$ are small
enough).  Due to its continuity near the positive real line,
$f_{\beta',k, s}(D)$ will map any such $z$ to a number with positive
real part provided $\epsilon$, and then $\delta$ and $\delta_\beta$,
are chosen to be small enough.  This proves part (b).

We now proceed to prove part~(a). The set
$C_2(\beta, \delta, d)$ is convex by design (for every $\delta > 0$),
so that by Lemma\nobreakspace \ref {lem:convexity-reduction}, we have
$F_{\beta',k,s}(D^k) = f_{\beta',k,s}(D)$ for all integers $k \geq 0$
and $s$ such that $k + \abs{s} \leq d$.  Thus it suffices to
show that $f_{\beta',k, s}(D) \subseteq D$ for $\delta$ and
$\delta_\beta$ small enough.  As in \MakeUppercase
Corollary\nobreakspace \ref{cor:set-contract}, we use the notation
$C(k) = q_\beta(C_2(\beta, \delta, k))$,
$C_1(k) = p_{\beta'}(C_2(\beta, \delta, k))$ (so that
$C_1(d) = \log(D)$) and
$\tilde{C}(k) = p_\beta(C_2(\beta, \delta, k))$, suppressing the
dependence of these sets on $\delta, \beta$ and $\beta'$ for
 simplicity of notation.

Recall from \protect \MakeUppercase {O}bservation\nobreakspace \ref {obv:ascend} that for every $\delta > 0$, we have the
following ascending chain of subsets:
\[\set{0} = C_2(\beta, \delta, 0) \subseteq C_2(\beta, \delta, 1)
  \subseteq \dots \subseteq C_2(\beta, \delta, k) \subseteq \dots
  \subseteq C_2(\beta, \delta, d).\] Since the $C(k)$, $C_1(k)$ and
$\tilde{C}(k)$ are obtained by applying the same maps (i.e., the maps
do not depend upon $k$) to these subsets, we therefore have similar
chains for these subsets as well:
\begin{gather*}
  \label{eq:16}
  \set{0} = C(0) \subseteq C(1) \subseteq \dots C(k) \subseteq \dots
  \subseteq C(d);\\
  \set{0} = C_1(0) \subseteq C_1(1) \subseteq \dots C_1(k) \subseteq
  \dots
  \subseteq C_1(d);\\
  \set{0} = \tilde{C}(0) \subseteq \tilde{C}(1) \subseteq \dots
  \tilde{C}(k) \subseteq \dots \subseteq \tilde{C}(d).
\end{gather*}
As observed earlier, $f_{\beta',k,s}(D) \subseteq D$ will follow if we
can show that $f^{\varphi}_{\beta',k,s}(C_1(d)) \subseteq C_1(d)$.  We
now prove this latter fact when $\delta$ and $\delta_\beta$ are small
enough positive constants.

We will do this by showing that, for small enough $\delta$ and
$\delta_\beta$, the following two facts hold for all $\beta'$ such
that $\abs{\beta' - \beta} \leq \delta_\beta$ and all integers
$k \geq 0$ and $s$ such that $k + \abs{s} \leq d$:
\begin{enumerate}
\item\label{item:10} $f_{\beta', k}^\varphi(C_1(d)) \subseteq C_1(k)$;
  and
\item\label{item:11} For any $z \in C_1(k)$, $z + s \log\beta'$ lies
  in $C_1(d)$.
\end{enumerate}
Together, these facts imply that
$f_{\beta', k, s}^\varphi(C_1(d)) = s\log\beta' + f_{\beta',
  k}^\varphi(C_1(d)) \subseteq C_1(d)$, which, as pointed out above,
establishes part~(a) of the theorem.  It therefore remains to prove
the two facts above for small enough $\delta,\delta_\beta$; we start
now with the proof of fact \eqref{item:10}.

Note first that fact \eqref{item:10} is trivially true when $k = 0$,
since in that case $f_{\beta', k}^\varphi$ is the constant map
$f_{\beta', k}^\varphi \equiv 0$.  We therefore assume in the proof of
fact \eqref{item:10} that $k \geq 1$.  Let $\delta_1$ be the smaller of the
corresponding quantities promised by Lemma\nobreakspace \ref
{lem:image-q} and Corollary\nobreakspace \ref {cor:set-contract}, and
let $\delta_2$ and $M$ be as given by Corollary\nobreakspace \ref
{cor:set-contract}.  Since
$q_\beta(C_2(\beta, \delta', d)) \subseteq q_\beta(C_2(\beta, \delta,
d))$ for $\delta' \leq \delta$, it follows from the last statement in
\protect \MakeUppercase {L}emma\nobreakspace \ref {lem:image-q}
that there exists an $\epsilon_0 > 0$ such that for all
$\delta \leq \delta_1$, $S_{\epsilon_0}^\dagger(C(d))$ is contained in
the interior of $C_0$.

Let $\tau$ and $\xi$ be positive constants satisfying
$\tau \leq \abs{\log{\beta}}$ and $\xi \leq 1/2$. Then, from
\MakeUppercase Lemma\nobreakspace \ref {lem:rect-contraction}, we know
that there exists positive constants $\eta < 1$ and $\nu$ such that for
any $\delta_\beta \leq \nu\min\inbr{\tau, \delta\xi}$ and any
$\beta' \in B_0$ such that $\abs{\beta - \beta'} \leq \delta_\beta$,
$f_{\beta', k}^\varphi$ is a map that
$(\frac{k}{d}(1 - \eta), \tau, \delta\xi)$-contracts rectangles in
$S_{\epsilon_0}^\dagger(C(d))$. Also, recall from \MakeUppercase
Lemma\nobreakspace \ref {lem:image-q} that
$C(k) = q_\beta\left( C_2(\beta,\delta, k) \right)$ is a rectangular
set for all $\delta > 0$ and positive integers $k$, and is given by
$C(k) = R(2k\abs{\log \beta}, k\delta)$.

Note now that for the rectangular set
$C(k) = R(2k\abs{\log \beta}, k\delta)$, $\tau < 2k\abs{\log \beta}$
and $\delta \xi < k \delta$.  Then, provided $\delta \leq \delta_1$
is chosen small enough such that $\epsilon \defeq M\delta^2$ satisfies
\begin{equation}
  \epsilon = M\delta^2 \leq \min\inp{\epsilon_0, \tau\eta/(4d), \xi\delta\eta/(4d)}/2,\label{eq:15}
\end{equation} we see from
\MakeUppercase Lemma\nobreakspace \ref {lem:chain-contraction} that
for every $\beta' \in B_0$ with $\abs{\beta' - \beta} \leq
\delta_\beta$ (with $\delta_\beta$ chosen as above depending upon the choice of  $\delta$), the $2^\C \rightarrow 2^\C$ map
$(\Stinv \circ {f_{\beta', k}^{\varphi}}\circ {\Stinv})$ is a map that 
$(\frac{k}{d}(1-\eta/2), \tau, \delta\xi)$-contracts rectangles in
$C(d)$. %
By our choice of $\tau$ and $\xi$, \MakeUppercase
Lemma\nobreakspace \ref {lem:closed-ideal} then applies to the
rectangular set $C(d)$ and the map
$g = \Stinv \circ f_{\beta', k}^\varphi
\circ \Stinv$, (with the quantity $\chi$ in the lemma
set to $\frac{k}{d}(1-\eta/2)$)  and implies that
\begin{equation}
(\Stinv \circ f_{\beta', k}^\varphi \circ \Stinv)(C(d)) \subseteq
  \frac{k}{d}C(d) = C(k).\label{eq:19}
\end{equation}
The latter in turn is equivalent to the statement that
 \begin{equation}
   f_{\beta', k}^\varphi(\Stinv(C(d))) \subseteq \St(C(k)).\label{eq:12}
 \end{equation}
 Now, from \MakeUppercase Corollary\nobreakspace \ref
 {cor:set-contract}, we know that if
 $\delta_\beta$ is chosen to be at most $\delta_2$, then for $\epsilon
 = M\delta^2$ as above, we have, for all non-negative integers $j \leq
 d$,
 \begin{equation}
   \St(C(j)) \subseteq C_1(j) \text{ and } C_1(j) \subseteq \Stinv(C(j)).\label{eq:17}
\end{equation}
Thus, if
$\delta$ is chosen small enough that eq.\nobreakspace \textup {(\ref
  {eq:15})} is satisfed and
$\delta_\beta$ is chosen to be at most $\min\inbr{\delta_2, \nu\tau,
  \nu\delta\xi}$ and also small enough that the
$\delta_\beta$ ball around $\beta$ is contained in
$B_0$, then for every $\beta'$ such that $|\beta' - \beta| \leq
\delta_\beta$, eqs.\nobreakspace \textup {(\ref {eq:12})}
and\nobreakspace \textup {(\ref {eq:17})} imply
\begin{displaymath}
  f_{\beta', k}^\varphi(C_1(d)) \subseteq f_{\beta', k}^\varphi(\Stinv(C(d)))
  \subseteq \St(C(k)) \subseteq C_1(k),
\end{displaymath}
which proves fact~(\ref {item:10}) above.

To prove fact~(\ref {item:11}), we note first that the claim follows
trivially when $s = 0$ (since $k \leq d$).  When $\abs{s}$ is
positive, $k + \abs{s} \leq d$ implies that $k \leq d-1$.  Therefore,
since $\beta$ is real and positive, any translate $s\log \beta + C(k)$
of the set $C(k) = R(2k\abs{\log\beta}, k\delta)$ is contained in
$R\inp{(2k + \abs{s})\abs{\log \beta}, k\delta} \subseteq R\inp{(d +
  k)\abs{\log \beta}, k\delta} \subseteq R\inp{(2d - 1)\abs{\log
    \beta}, (d-1)\delta}$.  The latter set is in turn contained in the
set
$S_{5\epsilon}(R(2d\abs{\log \beta}, d \delta))$ if $\delta$ is chosen
small enough that
$10\epsilon = 10M\delta^2 \leq \min\set{\delta, \abs{\log\beta}}$.  We
then have
$s\log \beta + C(k) \subseteq S_{5\epsilon}(C(d)) =
S_{3\epsilon}(S_{2\epsilon}(C(d)))$.  From eq.\nobreakspace \textup
{(\ref {eq:17})}, we then see that for $\delta_\beta$ and $\delta$
small enough and $\epsilon = M\delta^2$,
\begin{displaymath}
  S_{3\epsilon}^{\dagger}(s\log \beta + C(k)) \subseteq
  S_{2\epsilon}(C(d)) \subseteq C_1(d)
\end{displaymath}
and that
\begin{displaymath}
  s\log \beta + C_1(k) \subseteq S_{2\epsilon}^\dagger(s\log \beta + C(k)).
\end{displaymath}
Together, these imply that
$S_{\epsilon}^{\dagger}(s\log \beta + C_1(k)) \subseteq C_1(d)$.  Now,
if $\delta_\beta$ is chosen small enough, then, by the analyticity of
$\log$ in the closed $\delta_\beta$-ball around $\beta$,
we have that $\abs{\log \beta - \log{\beta'}} \leq \epsilon/d$ for any
$\beta'$ such that $\abs{\beta - \beta'} \leq \delta_\beta$. We then
have
$s\log \beta' + C_1(k) \subseteq S_{\epsilon}^{\dagger}(s\log \beta +
C_1(k))$.  We already showed that the latter is contained in $C_1(d)$.
This completes the proof of fact~(\ref {item:11}), and
hence of the theorem, when $\beta \neq 1$.

Finally, to handle the case $\beta = 1$, we proceed to define the sets
$C_1 = C_1(d)$ and $D = \exp(C_1)$ directly.  Consider the function
$h_\beta^\varphi(x) = \log\frac{\beta +e^x}{\beta e^x + 1}$.  By a direct
calculation, we have ${h_1^\varphi}(x) = 0$ for all $x$.  By the
continuity of ${h_\beta^\varphi}(x)$ in $\beta$ for all $x$ in a small
neighborhood of $0$, and the compactness of $R(\epsilon, \epsilon)$
for any $\epsilon > 0$, there then exists for any small enough such
$\epsilon$, a positive $\delta_\beta$ such that for all $\beta'$ with
$\abs{\beta' - 1} \leq \delta_\beta$ and any $x$ in
$R(\epsilon, \epsilon)$, we have
\begin{displaymath}
  \abs{\log \beta'} \leq M\epsilon^2 \text{ and }  |{{h_{\beta'}^\varphi}(x)}| \leq M\epsilon^2.
\end{displaymath}

Let $C_1 \defeq R(\epsilon, \epsilon)$ for such a small enough
$\epsilon < 1$ to be specified later, and let $\delta_\beta$ be chosen
as above given $\epsilon$ . Then, for \emph{any} integers $k$ and $s$
and any $\vec{x} \in C_1^{k}$, we have that
\begin{displaymath}
  \bigl|{F_{\beta', k, s}^\varphi(\vec{x})}\bigr| = \Bigl|{s\log \beta' + \sum_{i =
      1}^kh_{\beta'}^\varphi(x_i)}\Bigr| \leq \abs{s} M\epsilon^2 + kM\epsilon^2.
\end{displaymath}
The latter is less than $\epsilon$ when $k > 0$ and $s$ are such that
$k + \abs{s} \leq \Delta$, and provided $\epsilon$ is chosen small
enough that $\epsilon\Delta M$ is less than $1$.  We
then have $F_{\beta', k, s}^\varphi(C_1^k) \subseteq C_1$ for all
integers $k > 0$ and $s$ such that $k + \abs{s} \leq \Delta$, and this
proves both parts of the theorem with $D = \exp(C_1)$.

\appendix

\section{Weitz's self-avoding walk (SAW) tree construction}
\label{sec:weitzs-self-avoding}

In this section we briefly describe Weitz's self-avoiding walk (SAW)
tree construction used in the proof of Lemma~\ref{lem:weitz}.  We
refer to Refs.~\citenum{Weitz} and \citenum{zhaliabai09} for further
details.  We also note that a similar construction was used in 
Godsil's work on the monomer-dimer model.\cite{godsil_matchings_1981}

Given a graph $G = (V, E)$ and a vertex $v \in V$, Weitz's
construction produces a tree $T = \tsaw(G, v)$ (with some pinned
leaf vertices) which has the following property: if the root of the tree is
denoted $\rho$, then (in the notation introduced in
Section~\ref{sec:outline-proof})
\begin{displaymath}
  R_{G, v}(\beta) = R_{T, \rho}(\beta).
\end{displaymath}
We now describe how the tree $T$ is constructed from $G$ and $v$.
Consider all self-avoiding walks (i.e., those that do not revisit any
vertex) in the graph $G$ starting at the vertex $v$.  The set of these
walks has a natural tree structure, which we denote by $T'$. The root
$\rho$ of $T'$ is identified with the trivial zero-length
walk consisting of the vertex $v$ alone.  Inductively, if $\omega$ is
a node in $T'$ identified with the walk $\inp{v, v_1, v_2, \dots, v_l}$ of
length $l$, then the children of $\omega$ in $T'$ are identified with
walks of length $l + 1$ that are obtained by appending to $\omega$ a
neighbor (in $G$) of the last vertex $v_l$ in $\omega$ that does not
already appear in the walk $\omega$.

Consider now an arbitrary node in $T'$, and suppose that it is
identified with the walk $\omega = \inp{v, v_1, v_2, \dots, v_l}$.  Let
$L_\omega$ denote the set of neighbors of $v_l$ in $G$, except for its
immediate predecessor $v_{l-1}$ in $\omega$, that have already
appeared in $\omega$.  The tree $\tsaw(G, v)$ in Weitz's construction
is obtained from $T'$ as follows.  To each such node
$\omega = \inp{v, v_1, \dots, v_l}$ in $T'$, we attach as children the
(non-self-avoiding) walks obtained by appending to $\omega$ each of
the neighbors of $v_l$ in $L_\omega$, and then pinning these new nodes
according to a rule that we now describe.  Let $\mu$ be any
node in $T'$, and let $u$ be the last node in the self-avoiding walk
identified with $\mu$.  We choose arbitrarily a total order
$<_{\mu, u}$ on all neighbors (in $G$) of $u$.

Once such an order has been chosen for all nodes of $T'$, the pinning
of the newly added leaf nodes of $T = \tsaw(G, v)$ is decided as
follows.  As before, let $\omega$ be a node of $T'$ with last vertex
$v_l$, and consider the leaf vertex $\omega'$ of $T$ obtained by
appending to $\omega$ a neighbor $u$ of $v_l$ that already appears in
$\omega$.  Let $\mu$ be the node representing the prefix of the path
$\omega$ that ends at $u$, and let $u' \neq v_l$ be the successor of
$u$ in $\omega$.  Then, we pin $\omega'$ to `$+$' if
$v_l <_{\mu, u} u'$ and to `$-$' otherwise. We give an example of this
construction in Figure~\ref{fig:1}.  (In the example, the order at the
nodes of $T'$ is obtained by assigning a global order on the vertices
of $V$ according to their integer labels, and then using the induced
order on the subset of neighbors).  Note that if $G$ is of maximum
degree $\Delta = d + 1$, then all nodes of $T$, except for the root
node, have at most $d$ children, while the root may have up to
$\Delta$ children.

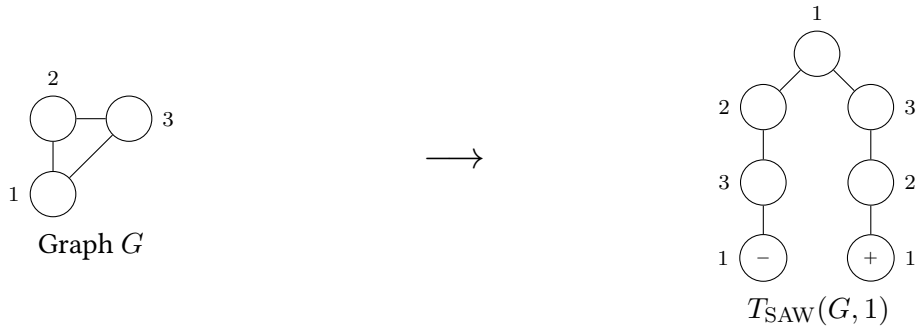
\begin{figure}[h]
  \centering
  \begin{minipage}[h]{0.4\linewidth}
    \centering
    \begin{tikzpicture}[main node/.style args =
      {#1/#2}{circle,draw,label=#1:{\scriptsize #2},minimum size=0.6cm}]
      \node[main node=left/{$1$}] (a) [circle] {};
      \node[main node=above/{$2$},above of=a] (b) [circle] {};
      \node[main node=right/{$3$},right of=b] (c) [circle] {};
      \draw (a) -- (b)
      (a) -- (c)
      (b) -- (c);
    \end{tikzpicture}

    Graph $G$
  \end{minipage}\begin{minipage}[h]{0.2\linewidth}
    \begin{center}
      {\Large $\longrightarrow$}
    \end{center}
  \end{minipage}\begin{minipage}[h]{0.4\linewidth}
    \centering
    \begin{tikzpicture}[main node/.style args =
      {#1/#2}{circle,draw,label=#1:{\scriptsize #2},minimum size=0.6cm}]
      \node[main node=above/{$1$}] (a) {};
      \node[main node=left/{$2$},below left of=a] (b1) {};
      \node[main node=left/{$3$},below of=b1] (c1) {};
      \node[main node=left/{$1$},below of=c1] (a1) {\tiny $-$};

      \node[main node=right/{$3$},below right of=a] (c2) {};
      \node[main node=right/{$2$},below of=c2] (b2) {};
      \node[main node=right/{$1$},below of=b2] (a2) {\tiny $+$};
      \draw (a) -- (b1) (a) -- (c2)
      (b1) -- (c1)
      (c1) -- (a1)
      (c2) -- (b2)
      (b2) -- (a2);
    \end{tikzpicture}

    $\tsaw(G, 1)$
  \end{minipage}
  \caption{Weitz's SAW tree construction}
  \label{fig:1}
\end{figure}

As stated above, Weitz's theorem\cite{Weitz,zhaliabai09} then shows
that
$  R_{G, v}(\beta) = R_{T, \rho}(\beta)$.
Since $T$ is a tree, $R_{T, \rho}(\beta)$ can be computed according to
a simple tree recurrence.  In particular, if the root $\rho$ of $T$
has $k$ children and $T_i$ is the subtree rooted at the $i$th child
$\rho_i$ of $\rho$ in $T$, for $1 \leq i \leq k$, then we have
(suppressing the dependence of all quantities on $\beta$ for ease of
notation)
\begin{align*}
  R_{T, \rho}
  &= \frac{Z_{T, \rho}^+}{Z_{T, \rho}^-}\\
  &= \frac{
    \prod_{i=1}^k(
    \beta Z_{T_i, \rho_i}^- + Z_{T_i, \rho_i}^+
    )}{
    \prod_{i=1}^k(
    \beta Z_{T_i, \rho_i}^+ + Z_{T_i, \rho_i}^-
    )}\\
  &=\prod_{i=1}^k\frac{\beta + R_{T_i,\rho_i}}{\beta R_{T_{i}, \rho_i}
    + 1} = F_{\beta, k, 0}(R_{T_1, \rho_1}, R_{T_2, \rho_2}, \dots, R_{T_k, \rho_k}).
\end{align*}
Here $F_{\beta, k, 0}$ is as defined in Lemma~\ref{lem:weitz}.  Note
also that the initial conditions for the recurrence (at the leaf nodes
of $T$) are given by $0$ (for leaf nodes pinned to $-$) and $\infty$
(for leaf nodes pinned to $+$).


\begin{thebibliography}{10}

\bibitem{anari2014kadison}
{\sc Anari, N., and Gharan, S.~O.}
\newblock The {Kadison}-{Singer} problem for strongly {Rayleigh} measures and
  applications to {Asymmetric TSP}.
\newblock In {\em Proc. 56th IEEE Symp. Found. Comp. Sci. (FOCS)\/} (Oct.
  2015).

\bibitem{anari_generalization_2017}
{\sc Anari, N., and Gharan, S.~O.}
\newblock A generalization of permanent inequalities and applications in
  counting and optimization.
\newblock In {\em Proc. 49th ACM Symp. Theory Comput. (STOC)\/} (June 2017),
  ACM, pp.~384--396.
\newblock \arxiv{1702.02937}.

\bibitem{bandyopadhyay_counting_2008}
{\sc Bandyopadhyay, A., and Gamarnik, D.}
\newblock Counting without sampling: {Asymptotics} of the log-partition
  function for certain statistical physics models.
\newblock {\em Random Structures \& Algorithms 33}, 4 (Dec. 2008), 452--479.

\bibitem{barvinok_computing_2015}
{\sc Barvinok, A.}
\newblock Computing the permanent of (some) complex matrices.
\newblock {\em Found. Comput. Math. 16}, 2 (2015), 329--342.

\bibitem{barvinok2017combinatorics}
{\sc Barvinok, A.}
\newblock {\em Combinatorics and Complexity of Partition Functions}.
\newblock Algorithms and Combinatorics. Springer, 2017.

\bibitem{BarvinokSoberon16b}
{\sc Barvinok, A., and Sober{\'{o}}n, P.}
\newblock Computing the partition function for graph homomorphisms with
  multiplicities.
\newblock {\em J. Combin. Theory, Ser. A 137\/} (2016), 1--26.

\bibitem{BarvinokSoberon16a}
{\sc Barvinok, A., and Sober{\'{o}}n, P.}
\newblock Computing the partition function for graph homomorphisms.
\newblock {\em Combinatorica 37}, 4 (Aug. 2017), 633--650.

\bibitem{bencs_note_2018}
{\sc Bencs, F., and Csikvári, P.}
\newblock Note on the zero-free region of the hard-core model, July 2018.
\newblock \arxiv{1807.08963}.

\bibitem{fisher65}
{\sc Fisher, M.~E.}
\newblock The nature of critical points.
\newblock In {\em Lecture notes in Theoretical Physics}, W.~E. Brittin, Ed.,
  vol.~7c. University of Colorado Press, 1965, pp.~1--159.

\bibitem{georgii88:_gibbs_measur_phase_trans}
{\sc Georgii, H.-O.}
\newblock {\em Gibbs Measures and Phase Transitions}.
\newblock De Gruyter Studies in Mathematics. Walter de Gruyter Inc., October
  1988.

\bibitem{godsil_matchings_1981}
{\sc Godsil, C.~D.}
\newblock Matchings and walks in graphs.
\newblock {\em J. Graph Theory 5}, 3 (Sept. 1981), 285--297.

\bibitem{harvey2016computing}
{\sc Harvey, N. J.~A., Srivastava, P., and Vondr{\'a}k, J.}
\newblock Computing the independence polynomial: from the tree threshold down
  to the roots.
\newblock In {\em Proc. 29th ACM-SIAM Symp. Discrete Algorithms (SODA)\/}
  (2018), pp.~1557--1576.
\newblock \arxiv{1608.02282}.

\bibitem{isi25}
{\sc {Ising}, E.}
\newblock {Beitrag zur Theorie des Ferromagnetismus}.
\newblock {\em Z. Phys. 31\/} (Feb. 1925), 253--258.

\bibitem{kim_partition_2008}
{\sc Kim, S.-Y., Hwang, C.-O., and Kim, J.~M.}
\newblock Partition function zeros of the antiferromagnetic {Ising} model on
  triangular lattice in the complex temperature plane for nonzero magnetic
  field.
\newblock {\em Nucl. Phys. B 805}, 3 (Dec. 2008), 441--450.

\bibitem{lee_statistical_1952}
{\sc Lee, T.~D., and Yang, C.~N.}
\newblock Statistical theory of equations of state and phase transitions. {II.}
  {Lattice} gas and {Ising} model.
\newblock {\em Phys. Rev. 87}, 3 (1952), 410--419.

\bibitem{li_correlation_2011}
{\sc Li, L., Lu, P., and Yin, Y.}
\newblock Correlation decay up to uniqueness in spin systems.
\newblock In {\em Proc. 24th ACM-SIAM Symp. Discrete Algorithms (SODA)\/}
  (2013), pp.~67--84.

\bibitem{liu2017ising-full}
{\sc Liu, J., Sinclair, A., and Srivastava, P.}
\newblock The {Ising} partition function: {Zeros} and deterministic
  approximation.
\newblock In {\em Proc. 58th Annual IEEE Symp. Found. Comp. Sci. (FOCS)\/}
  (2017), pp.~986--997.
\newblock \arxiv{1704.06493}.

\bibitem{lu_density_2001}
{\sc Lu, W.~T., and Wu, F.~Y.}
\newblock Density of the {Fisher} {Zeroes} for the {Ising} {Model}.
\newblock {\em J. Stat. Phys. 102}, 3-4 (Feb. 2001), 953--970.

\bibitem{lyons_ising_1989}
{\sc Lyons, R.}
\newblock The {Ising} model and percolation on trees and tree-like graphs.
\newblock {\em Commun. Math. Phys. 125}, 2 (1989), 337--353.

\bibitem{mann_approximation_2018}
{\sc Mann, R.~L., and Bremner, M.~J.}
\newblock Approximation algorithms for complex-valued {Ising} models on bounded
  degree graphs, June 2018.
\newblock \arxiv{1806.11282}.

\bibitem{marcus_interlacing_2015}
{\sc Marcus, A., Spielman, D., and Srivastava, N.}
\newblock Interlacing families {I}: {Bipartite} {Ramanujan} graphs of all
  degrees.
\newblock {\em Ann. Math.\/} (July 2015), 307--325.

\bibitem{marcus2015interlacing}
{\sc Marcus, A.~W., Spielman, D.~A., and Srivastava, N.}
\newblock Interlacing families {II}: {Mixed} characteristic polynomials and the
  {Kadison}-{Singer} problem.
\newblock {\em Ann. Math. 182\/} (2015), 327--350.

\bibitem{patel_deterministic_2016}
{\sc Patel, V., and Regts, G.}
\newblock Deterministic polynomial-time approximation algorithms for partition
  functions and graph polynomials.
\newblock {\em SIAM J. Comput. 46}, 6 (Dec. 2017), 1893--1919.
\newblock \arxiv{1607.01167}.

\bibitem{peters17:_sokal}
{\sc Peters, H., and Regts, G.}
\newblock On a conjecture of {Sokal} concerning roots of the independence
  polynomial, Jan. 2017.
\newblock \arxiv{1701.08049}.

\bibitem{ScottSokal}
{\sc Scott, A., and Sokal, A.}
\newblock {The repulsive lattice gas, the independent-set polynomial, and the
  Lov\'{a}sz local lemma}.
\newblock {\em J. Stat. Phys. 118}, 5-6 (2004), 1151--1261.

\bibitem{Shearer}
{\sc Shearer, J.~B.}
\newblock On a problem of {S}pencer.
\newblock {\em Combinatorica 5}, 3 (1985).

\bibitem{sinclair_approximation_2012}
{\sc Sinclair, A., Srivastava, P., and Thurley, M.}
\newblock Approximation algorithms for two-state anti-ferromagnetic spin
  systems on bounded degree graphs.
\newblock {\em J. Stat. Phys. 155}, 4 (2014), 666--686.

\bibitem{sly12}
{\sc Sly, A., and Sun, N.}
\newblock Counting in two-spin models on $d$-regular graphs.
\newblock {\em Ann. Probab. 42}, 6 (Nov. 2014), 2383--2416.

\bibitem{straszak_real_2016}
{\sc Straszak, D., and Vishnoi, N.~K.}
\newblock Real stable polynomials and matroids: Optimization and counting.
\newblock In {\em Proc. 49th ACM Symp. Theory Comput. (STOC)\/} (June 2017),
  ACM, pp.~370--383.
\newblock \arxiv{1611.04548}.

\bibitem{Weitz}
{\sc Weitz, D.}
\newblock Counting independent sets up to the tree threshold.
\newblock In {\em Proc. 38th ACM Symp. Theory Comput. (STOC)\/} (2006),
  pp.~140--149.

\bibitem{leeyan52}
{\sc Yang, C.~N., and Lee, T.~D.}
\newblock Statistical theory of equations of state and phase transitions. {I.
  Theory} of condensation.
\newblock {\em Phys. Rev. 87}, 3 (Aug. 1952), 404--409.

\bibitem{zhaliabai09}
{\sc Zhang, J., Liang, H., and Bai, F.}
\newblock Approximating partition functions of the two-state spin system.
\newblock {\em Inf. Process. Lett. 111}, 14 (2011), 702--710.

\end{thebibliography}
\end{document}